\documentclass[sigplan,10pt]{acmart}

\acmConference{\hspace{-23.5cm}Pre-print}{\hspace{22.3cm}Pre-print}{March 2018}
\acmYear{\hspace{-5cm}}
\acmISBN{} 
\acmDOI{} 
\startPage{1}

\setcopyright{none}

\usepackage{booktabs}   
\usepackage{subcaption} 

\usepackage[all]{xy} 
\usepackage{listings}
\usepackage{amsmath, mathrsfs, amsfonts, amssymb,nicefrac}
\usepackage{stmaryrd}
\usepackage{enumerate}
\usepackage{multirow}
\usepackage{multicol}
\usepackage{array}
\usepackage[all]{xy}
\usepackage{tikz}
\usepackage{pgfplots}
\usepgfplotslibrary{patchplots}
\usepgfplotslibrary{fillbetween}
\usetikzlibrary{snakes}
\usetikzlibrary{patterns}

\newcommand{\id}{\mathrm{id}}  
  
\newcommand{\inc}{\hookrightarrow}

\newcommand{\prog}[1][p]{\mathtt{#1}}
\newcommand{\lsem}{\llbracket}
\newcommand{\rsem}{\rrbracket}

\newcommand{\maybe}{\prog[Maybe]}
\newcommand{\sComp}{\hspace{1pt}\large{\mathbf{;}}\hspace{2pt}}
\newcommand{\At}{\prog[At]}
\newcommand{\pfs}{\hspace{0.2cm}}

\newcommand{\cat}[1][C]{\mathbf{#1}}
\newcommand{\Set}{\mathbf{Set}}

\newcommand{\Pol}{\mathbf{Pol}}
\newcommand{\Kl}[1][T]{\cat_{#1}}

\newcommand{\Vect}{\mathbf{Vect}}

\newcommand{\El}[1][T]{\mathbf{El}(#1)}
\newcommand{\catC}{\mathbf{C}}

\newcommand{\N}{\mathbb{N}}

\newcommand{\R}{\mathbb{R}}
\newcommand{\Q}{\mathbb{Q}}
\newcommand{\unit}{\left[0,1\right]}
\newcommand{\EK}[1][T]{\mathrm{End}_{#1}}
\newcommand{\Var}{\mathcal{V}}

\newcommand{\op}{^{\mathrm{op}}}
\newcommand{\colim}{\operatornamewithlimits{colim}}
\newcommand{\Lan}[1][\Inc]{\mathrm{Lan}_{#1}}

\newcommand{\Id}{\mathsf{Id}}
\newcommand{\Pow}{\mathsf{P}}
\newcommand{\NPow}{\mathsf{Q}}

\newcommand{\Giry}{\mathsf{G}}
\newcommand{\Maybe}{\mathsf{M}}

\newcommand{\GMset}[1][S]{\mathsf{B}_{#1}}
\newcommand{\Forg}{\mathsf{U}}
\newcommand{\Free}{\mathsf{F}}
\newcommand{\Dist}{\mathsf{D}}
\newcommand{\HybM}{\mathsf{H}}
\newcommand{\Yoneda}{\mathsf{Y}}
\newcommand{\Diag}{\mathscr{D}}
\newcommand{\Inc}{\mathsf{I}}

\newcommand{\one}{\underline{1}}

\newcommand{\inv}{^{-1}}
\newcommand{\supp}{\mathrm{supp}}
\newcommand{\ari}{\mathrm{ar}}
\newcommand{\const}[1]{\underline{#1}}
\newcommand{\conc}{\mathbin{+\mkern-08mu+}}
\newcommand{\comp}{\circ}

\newcommand{\sps}[1]{\hspace{-1pt}#1\hspace{-1pt}}
\newcommand{\ssps}[1]{\hspace{-2pt}#1\hspace{-2pt}}
\newenvironment{myProof}[1]
	{
	\noindent \textbf{Proof of #1.}\\
	}
	{
	\begin{flushright}$\blacksquare$\end{flushright}
	}
	
\usepackage[colorinlistoftodos]{todonotes}
\presetkeys{todonotes}{fancyline, color=blue!30, size=\small}{}
\def\red#1{\textcolor{red}{#1}}

\begin{document}

\title[Compositional semantics for new paradigms]{Compositional semantics for new paradigms: probabilistic, hybrid and beyond}       


\author{Fredrik Dahlqvist}
\affiliation{
  \institution{University College London}            
}
\email{f.dahlqvist@ucl.ac.uk}          

\author{Renato Neves}
\affiliation{
  \institution{University of Minho}           
}
\email{nevrenato@di.uminho.pt}         

\begin{abstract}
%

  Emerging computational paradigms, such as probabilistic and hybrid
  programming, introduce new primitive operations that often need to
  be combined with classic programming constructs. However, it still
  remains a challenge to provide a semantics to these features and
  their combination in a systematic manner.

  For this reason, we introduce a \emph{generic}, \emph{monadic}
  framework that allows us to investigate not only which programming
  features a given paradigm supports, but also on how it can be
  extended with new constructs.  By applying our method to the
  probabilistic and hybrid case, we list for example all binary
  program operations they possess, and show precisely when and if
  important axioms such as commutativity and idempotency hold.  Using
  this framework, we also study the possibility of incorporating
  notions of failure and non-determinism, and obtain new results on
  this topic for hybrid and probabilistic programming.
  
\end{abstract}

\ccsdesc[500]{Software and its engineering~General programming languages}
\ccsdesc[300]{Social and professional topics~History of programming languages}

\keywords{Probabilistic program, hybrid program, monad,
  semantics} 

\maketitle

\setlength{\abovedisplayskip}{2pt}
\setlength{\belowdisplayskip}{1pt}
\setlength{\dbltextfloatsep}{0pt}
\setlength{\textfloatsep}{0pt}
\setlength{\dblfloatsep}{0pt}
\setlength{\floatsep}{0pt}
\setlength{\parskip}{2pt}

\section{Introduction}

Probabilistic programming languages such as Church
\cite{goodman2012church}, Anglican \cite{wood-aistats-2014} or
Probabilistic C \cite{pmlr-v32-paige14} have become increasingly
popular in the last years, and although progress has been made in
developing semantics for them many questions remain. In particular,
\emph{how does one interpret combinations of probabilistic features
  with `classical' features like error handling or non-deterministic
  choice}?  Consider for example the program below, written in
Probabilistic C-style.
\begin{lstlisting}[basicstyle=\fontsize{8}{10}\ttfamily]
int main() {
  int a;
  int c_1=bernoulli(0.3);
  int c_2=bernoulli(0.6);
  printf("Please input an integer: ");
  scanf("%d", &a);
  if(a % 3 == 0){
    return c_1;
  } else if (a % 3 == 1){
  return c_2;
  } else {


  exit(EXIT_FAILURE);
  }
}
\end{lstlisting}
\vspace{-0.1cm}
\noindent
This program non-deterministically combines two probabilistic
instructions -- producing Bernoulli trials which return 1 with
probability 0.3 (resp. 0.6) and 0 with probability 0.7 (resp. 0.4) --
with an execution failure.  By abstracting away from the C-style
grammar and moving to an algebraic syntax, we want to understand how
to interpret the expression
\[
\prog[(1+_{.3}0)+(1+_{.6}0)+abort]
\] 
where $\prog[+_\lambda]$ is the binary probabilistic choice operator
with parameter $\lambda\in\unit$. It is easy enough to provide a
semantics to the purely probabilistic instructions in terms of Markov
kernels $\N\to \Dist\N$, where $\Dist$ is the finitely supported
distribution monad. \emph{But can we interpret $\prog[+]$ and
  $\prog[abort]$ in this semantics? And if not, how can we modify
  $\Dist$ to support these constructs?} We aim to provide firm answers to this kind of question.

The challenge that we described above is not unique to probabilistic
programming. Parallel to the latter, recent years have witnessed a
flurry of research activity aiming to formalise the programmable
features of hybrid systems \cite{hofner09,platzer10,suenaga11}, which
require an orchestrated use of both \emph{classic program constructs}
and \emph{systems of differential equations}.  Consider, for example,
the `C-style' hybrid program below.
  \begin{lstlisting}[mathescape=true,commentstyle=\color{gray},
    basicstyle=\fontsize{8}{10}\ttfamily]
int cool_or_heat() {
  int a;
  printf("Please input an integer: ");
  scanf("%d", &a);
  if(a == 0){
    // Heating up
    (dtemp = 1 & 3);
    return 0; // success 
  } else if (a == 1){
    // Cooling down 
    (dtemp = -1 & 3);
    return 0; // success 
  } else {
    exit(EXIT_FAILURE);
  }
}
\end{lstlisting}
\noindent
Depending on the input, the program increases or decreases a reactor's
temperature during three miliseconds or aborts: there exists a global
variable (\texttt{temp}) that registers the current temperature, the
expression \texttt{(dtemp = 1 \& 3)} dictates how the temperature is
going to evolve for the next three miliseconds, and similarly for
\texttt{(dtemp = -1 \& 3)}. Abstracting from the C-style
grammar, we want to interpret the expression
\[
\prog[(dtemp = 1 \> \& \> 3)+(dtemp = -1 \> \& \> 3)+abort]
\]
in a suitable manner, and this yields questions completely analogous
to the probabilistic case. Purely hybrid programs can be naturally
interpreted as functions $\R^n \to \HybM(\R^n)$ where $\HybM$ is the
hybrid monad \cite{neves15}. \emph{But it is not clear how to extend
  this semantics in a systematic manner so that it incorporates
  exceptions and non-determinism.}

\vspace{-2ex}
\paragraph{Approach and contributions.} The paper is divided in two
halves: in the first part we develop a framework -- \emph{Kleisli
  representations} -- so that we can analyse program semantics in a
systematic way. In the rest of the paper, we show how program semantics
can be explicitly built in this framework, and use it to analyse
the hybrid and probabilistic paradigms.

The framework is laid out in Section ~\ref{sec:KlRep}. It reinterprets
Moggi's idea \cite{1991:MoggiNotions} of interpreting a program $\prog$ as a
Kleisli arrow $\lsem\prog\rsem: X\to TX$ and sequential composition as
Kleisli composition by saying that the interpretation map is a
\emph{monoid morphism} from the monoid of programs to the monoid of
\emph{endomorphisms of $X$ in the Kleisli category of $T$}.

Representing an algebraic structure as a collection of endomorphisms
is an idea very familiar to mathematicians and physicists
alike, it sits at the heart of \emph{Representation Theory}, a vast
field of research which has extensively studied the representations of
groups and Lie algebras as endomorphisms of vector spaces, with
applications ranging from the classification of finite groups to
quantum field theory
\cite{steinberg2011representation,haywood2011symmetries}. In this work
we re-interpret denotations \textit{\`{a} la} Moggi as
\emph{representations} in the usual mathematical understanding of the
word. We will refer to a representation in the Kleisli category $\Kl$
of $T$ as a \emph{Kleisli $T$-representation}.

Representation theory provides a natural and useful way of
thinking about program semantics: it captures both the algebraic
aspects of the language and the coalgebraic aspects of its
interpretation in a simple, well-known mathematical object, it separates the role of sequential composition from other
operations on programs, it connects smoothly with the existing
literature on algebraic effects
\cite{1999:PowerEnriched,2001:PlotkinPowerAEAdequacy,2001:PlotkinPowerSemanticsAE,2003:PlotkinPowerEffects},
and leads to new methods for building program semantic.

Whilst the Kleisli representation of sequential composition boils down
to a natural transformation $T\circ T\to T$, all other binary
operations on programs will be interpreted as natural transformations
$T\times T\to T$, an idea first proposed by Plotkin and Power
\cite{1989:MoggiComputational,1991:MoggiNotions}. Thus, in order to
determine precisely which binary program operations are supported by a
monad, we must completely enumerate all natural
transformations $T\times T\to T$. Section~\ref{sec:operations},
introduces a set of techniques to do this for a large list of
well-known monads. For example, we list all binary operations for the
hybrid, probabilistic, non-deterministic, and partial paradigms.

Our next step is to study program axiomatics, i.e.\@ the
(im)possibility for a monad to support program operations satisfying
some given axioms. This is done in Section \ref{sec:axioms} where we
focus on commutativity, idempotence, units, and absorption. We give
fine grained results on which axioms can be supported by Kleisli
$T$-representations for different monads, which clarifies the types of
operation specific computational paradigms support. We show for
example that the hybrid paradigm does not admit a non-deterministic
choice, that the probabilistic paradigm supports precisely one
commutative idempotent operation and that it does not support failures.

Section \ref{sec:combining} shows that combining monads with 
  the $\maybe$ monad $\Maybe$ or the powerset monad $\Pow$ (i) provides a
  generic interpretation of tests, and (ii) yields richer monads with
  which to overcome representability issues highlighted in Section
  \ref{sec:axioms}. We strengthen a well-known result from
  \cite{2003:VaraccaProbability} by showing that there cannot exist
  \emph{any} monad structure on $\Pow\Dist$ whatsoever. We also prove
  the existence of a distributive law of the non-empty powerset  monad $\NPow$
  over the hybrid one. This allows to generate powerful, hybrid programming
  languages that mix non-deterministic assignments with differential
  predicates.  

  We assume basic knowledge of category theory and monads. All 
  proofs can be found in the appendix.
\vspace{-2ex}
\section{Building semantics for hybrid and probabilistic programs}
\label{sect:illustr}
Let us illustrate the questions raised in the introduction and some of
the solutions that we have developed for them by looking at two
emerging programming paradigms.

\vspace{-1ex}
\paragraph{Hybrid programs} The distinguishing feature of hybrid programming is that it emphasises
and makes explicit the interaction between digital devices and
physical processes. This
is becoming essential for the software engineer, because he/she
needs more and more often to develop complex systems that are deeply
intertwined with physical systems
\cite{tabuada09,platzer10}, e.g.\@ cruise
controllers, thermostats, etc.

Let us build a very simple hybrid programming language. Take a finite
set of real-valued variables $X = \{ x_1, \dots, x_n\}$ and denote by
$\At(X)$ the set, given by the grammar
\begin{flalign*}
  & \varphi = (x_1 := t, \dots, x_n := t) \mid (\dot{x}_1 = t,
  \dots, \dot{x}_n = t \> \& \> r) \\ & t = r \mid r
  \cdot x \mid t + t
\end{flalign*}
where $r$ is a real number and $x \in X$. Then define $\prog[Hyb(X)]$
as the free monoid over $\At(X)$ given by the grammar
\[
\prog[p]=\prog[a\in At(X) \mid skip\mid p\sComp p]
\]
and with the usual monoidal laws. One possible program is the
composition
$\prog[a := 10 \sComp (] \dot{\prog[p]}\prog[ \> = v,] \> \>
\dot{\prog[v]}\prog[ \> = a \> \& \> 3)]$, which, intuitively, sets the
acceleration of a vehicle to $\prog[10] m/s^2$ and then makes it move
during three milliseconds.

In Section~\ref{sec:KlRep}, we will see that the semantics for this
language comes naturally as a Kleisli representation
\begin{flalign*}
  \lsem - \rsem : \prog[Hyb(X)] \to \EK[\HybM](\R^3)
\end{flalign*}
where $\HybM$ is the hybrid monad \cite{neves15} (whose definition is recalled
in the following section). We now want to endow
$\prog[Hyb(X)]$ with other programming features, such as
$\prog[abort]$ operations and non-deterministic choice.  In the framework of Klesili
representations, this amounts to providing $\EK[\HybM](\R^3)$ a
suitable algebraic structure that supports these constructs. Consider,
for example, the language $\prog[Hyb^0(X)]$ with syntax
\[
\prog[p]=\prog[a\in At(X) \mid skip\mid p\sComp p \mid 0]
\]
Our goal is to build a Kleisli
representation
\begin{flalign*}
  \lsem - \rsem : \prog[Hyb^0(X)] \to \EK[\HybM](\R^3)
\end{flalign*}
for this language.  Corollary~\ref{thm:TNfromInitCriterion}
shows that this is impossible, because there does \emph{not} exist a
natural transformation $T^{\prog[0]}:=\one\to \HybM$ to interpret
$\prog[0]$.  In other words, pure hybrid computations do \emph{not}
support $\prog[abort]$ operations.

In order to surpass this obstacle, one may want to consider the
partial hybrid monad $\HybM(- + 1)$. We show that this monad has
precisely one natural transformation $T^{\prog[0]}:=\one\to \HybM(-+1)$
to interpret $\prog[0]$ such that the axiom
$\prog[0 \sComp p] = \prog[0]$ holds. In contrast to classic
paradigms, we prove that partial hybrid programs do not admit the
axiom $\prog[p \sComp 0] = \prog[0]$ and such is to be expected.

Suppose now that we wish to extend the language $\prog[Hyb(X)]$ with a
non-deterministic choice $(+)$. As already mentioned, in our framework
this amounts to finding a suitable natural transformation
$\HybM \times \HybM \to \HybM$. Using the set of techniques introduced
in the paper, we list all such transformations and quickly discover
that none of them is commutative. This means that in pure hybrid
computations one cannot expect a non-deterministic choice operation.

To solve this issue, we show that there exists a distributive law of
the non-empty powerset monad $\NPow$ over $\HybM$.  The monad
$\NPow \HybM$ inherits the natural transformation
$\NPow \HybM \times \NPow \HybM \to \NPow \HybM$ that takes unions and
this allows to extend the language $\prog[Hyb(X)]$ with a
non-deterministic choice operation in the usual way.

\vspace{-1ex}
\paragraph{Probabilistic programs}
Consider the simple probabilistic programming language $\prog[Prob]$
described by the syntax
\begin{align*}
\prog[p]=\prog[a\in At\mid skip \mid p\sComp p \mid p +_\lambda p] \hspace{1cm}\lambda\in\unit\cap\Q
\end{align*}
where $+_\lambda$ is the probabilistic choice operation and
$\prog[At]$ is a set of atomic programs.  $\prog[Prob]$
satisfies the following axioms:

\begin{center}
\begin{tabular}{r@{\hskip 3pt} l@{\hskip 15pt}  r@{\hskip 3pt}  l}
(1)& $\prog[p\sComp skip \sps{=} skip\sComp p\sps{=} p]$ & (5)&  $\prog[p\sComp (q\sComp r)\sps{=} (p\sComp q)\sComp r]$ \\
(2)& $\prog[p\sComp (q+_\lambda r)\sps{=} (p\sComp q)+_\lambda (p\sComp r)]$ & (6) & $\prog[p+_\lambda q\sps{=} q+_{1-\lambda}p]$\\
(3)& $\prog[(p+_\lambda q)\sComp r=(p\sComp r) \sps{+_\lambda} (q\sComp r)]$ & (7) &  $\prog[p +_\lambda p\sps{=}  p]$\\
(4)& \multicolumn{3}{l}{$\prog[p+_\lambda(q+_{\tau} r)\sps{=}  (p+_{\hspace{-3pt}\frac{\lambda}{\lambda+(1-\lambda)\tau}}\hspace{-2pt}q)+_{\lambda+(1-\lambda)\tau} r]$}
\end{tabular}
\end{center}

\noindent
Let us call an algebraic structure for this signature and these
equations a \emph{convex semiring}. $\prog[Prob]$ is the free convex semiring over $\prog[At]$. The $+_\lambda$ fragment
of such a structure is known as a \emph{convex algebra} (modulo an
extension to $n$-ary affine combinations, see
e.g. \cite{2017:Sokolova}), and convex algebras
are precisely the Eilenberg-Moore algebras for the distribution monad
$\Dist:\Set\to\Set$. The set $\EK[\Dist](X)$
can be equipped with a convex algebra structure inherited from
$\Dist$ in the obvious way. In fact,
$\EK[\Dist](X)$ is also a convex semiring, i.e. all the axioms
listed above also hold when $\prog[\sComp]$ is interpreted as the
Kleisli composition $\circ_\Dist$.

We now interpret the language above in terms of rational Markov
kernels (i.e. maps $X\to \Dist_r X$ where $\Dist_r$ is the monad of
\emph{rational} probability distributions), by defining Kleisli
$\Dist_r$-representations as follows: for each $\prog[a\in At]$ choose
an interpretation $\lsem \prog[a]\rsem: X\to\Dist_r X$ for some state
space $X$. The semantics is then extended inductively:
\vspace{-2ex}
\begin{center}
\begin{tabular}{l l l }
$\lsem \prog[skip]\rsem \ssps{=} \eta_X$ &  $\lsem \prog[p\sComp q]\rsem\ssps{=}\lsem \prog[q]\rsem
  \sps{\circ} \lsem \prog[p]\rsem$  &  $\lsem \prog[p\sps{+_\lambda} q]\rsem\ssps{=}\lambda\lsem p \rsem \sps{+}
  (1\sps{-}\lambda)\lsem \prog[q]\rsem$ 
\end{tabular}
\end{center}

\noindent
This Kleisli $\Dist_r$-representation
$\lsem -\rsem: \prog[Prob]\to \EK[\Dist_r](X)$ is a convex semiring
homomorphism.  We now want to extend the language $\prog[Prob]$ with
other operations on programs, such as non-deterministic choice (as in ProbNetKAT \cite{2016:ProbNetKAT,2017:CantorScott}), iteration, or parallel composition. Again, non-deterministic choice is enough to illustrate some
goals and contributions of this paper. Consider the language
$\prog[Prob^+]$ whose syntax is
\[
\prog[p]=\prog[a\in At\mid skip\mid 0 \mid p\sComp p \mid p +_\lambda p\mid p+p] \hspace{1.7em}\lambda\in\unit\cap\Q
\] 
and whose axioms are those of $\prog[Prob]$ together with those making
$(\prog[;,1,+,0])$ an idempotent semiring.

The most obvious strategy to provide a semantics for $\prog[Prob^+]$ is
to try to put a $+$ operation satisfying the axioms above directly on
$\EK[\Dist_r](X)$ and then define Kleisli representations
$\prog[Prob^+]\to\EK[\Dist_r](X)$ in the same way as we did for
$\prog[Prob]$. Since $+$ will \textit{in fine} be interpreted by
a natural transformation $\Dist_r\times \Dist_r\to\Dist_r$, our first
task is to characterise these.
We show in Theorem \ref{thm:TNatRatDist} that they are precisely
the convex sum operations. It follows that by choosing the equally
weighted convex sum $+_{\nicefrac{1}{2}}$ we can equip $\EK[\Dist_r](X)$
with an operation which satisfies all the axioms above apart from (10)
and (11) -- actually, we will see that this is the only possible
choice. The axioms (10) and (11) fail because once again there are \emph{no} natural transformations
$T^0:=\one\to \Dist_r$ to interpret $\prog[0]$
(Corollary~\ref{thm:TNfromInitCriterion}).  In particular there cannot
exist \emph{any} Kleisli $\Dist_r$-representations of
$\prog[Prob^+]$. 
To overcome this obstacle we will, as in the hybrid case, 
consider a more complex monad in Section \ref{sec:combining} -- $\Dist\Maybe$ -- which adds the missing notion of
partial computation to $\Dist$, and restrict our attention to a particular class of non-deterministic instruction.


\section{Kleisli representations}\label{sec:KlRep}

\paragraph{Kleisli representation of monoids.}
  Given a monoid $(M,\cdot,1)$, a monad $T:\cat\to\cat$ and a
  $\cat$-object $X$, define a \emph{Kleisli $T$-representation of
    $M$ in $X$}, or simply a \emph{Kleisli representation of $M$ in
    $X$} if there is no ambiguity, as a monoid homomorphism
\[
  \rho: (M,\cdot,1)\to (\EK(X), \circ_T, \eta_X^T)
\]
where $\circ_T$ is Kleisli composition and $\eta^T$ is the unit
of $T$.

\begin{example}[Classical linear representations]
 Let $\Free: \Set\to\Vect$ be the functor building free vector
  spaces over some chosen field, and let $\Forg: \Vect\to\Set$ be the
  corresponding forgetful functor. The composition
  $\Forg\Free:\Set\to\Set$ is a monad, and for any group $G$, a
  Kleisli $\Forg\Free$-representation of $G$ on a finite set $n$ is
  simply the usual notion of linear representation of $G$ on the
  $n$-dimensional vector space.
\end{example}


\begin{example}[Stochastic processes]
  The category $\Pol$ is the category of Polish spaces, i.e.\@
  separable, completely metrisable topological spaces, and continuous
  maps. The Giry monad \cite{1981:Giry} $\Giry:\Pol\to\Pol$ associates
  to every Polish space $X$ the set of probability distributions on
  $X$ together with the topology of weak convergence, which is
  Polish. On morphisms, it associates to any continuous map
  $f: X\to Y$ the map
  $\Giry f:\Giry X\to\Giry Y, \mu\mapsto f_\ast(\mu)$ taking the
  pushforward of measures. A Kleisli $\Giry$-representation of the
  monoid $([0,\infty),+,0)$ of non-negative reals is a
  \emph{stochastic process.}
\end{example}

\begin{example}[Hybrid systems]
The functorial part of the hybrid monad  
 $\HybM : \Set \to \Set$ \cite{neves15}
is defined by
\begin{flalign*}
  \HybM = \coprod_{d \in [0,\infty)} \hom([0,d], -)
\end{flalign*}
Its unit
is given by the equation
$\eta_X ( x ) = (\const{x}, 0)$, with $\const{x}$
the constant function on $x$, and the multiplication is defined
by
\begin{flalign*}
  \mu_X (f,d) = (\theta_{X} \circ f, d) \conc (f (d))
\end{flalign*}
with $\theta : \HybM \to \Id$ the natural transformation that sends an
evolution $(f,d)$ to $f(0)$ and $\conc : \HybM \times \HybM \to \HybM$
the natural transformation that concatenates two
evolutions. Intuitively, the multiplication will be used to
concatenate the evolutions produced by two hybrid programs.

Recall from Section~\ref{sect:illustr} the grammar $\At(X)$ of atomic
hybrid programs and denote the usual interpretation of a term $t$ over
a valuation $(v_1,\dots,v_n) \in \R^n$ by
$\lsem t \rsem_{(v_1,\dots,v_n)}$ or simply $\lsem t \rsem $ if the
valuation is clear from the context. Since linear systems of ordinary
differential equations always have unique solutions
\cite{perko2013differential}, there exists an interpretation map
\begin{flalign*}
  \At(X) \to \EK[\HybM](\R^n)
\end{flalign*}
that sends $(x_1 := t_1, \dots, x_n := t_n)$ to the function
$\R^n \to \HybM(\R^n)$ defined by,
\begin{flalign*}
  (v_1,\dots,v_n) \mapsto \eta_{\R^n} \left (\lsem {t_1} \rsem ,\dots ,\lsem {t_n} \rsem
  \right )
\end{flalign*}
and that sends $(\dot{x}_1 = t_1, \dots, \dot{x}_n = t_n \> \& \> d)$
to the respective solution $\R^n \to (\R^n)^{[0,\infty)}$ but
restricted to $\R^n \to (\R^n)^{[0,d]}$. The free monoid extension of
this interpretation map provides a Kleisli representation
\begin{flalign*}
  \prog[Hyb](X) \to \left (\EK[\HybM](\R^n), \comp,
    \eta_{\R^n} \right )
\end{flalign*}
which includes both assignments and differential equations. In the
appendix we provide more details about this language, and give
examples of other languages generated by the hybrid monad.
\end{example}

\paragraph{Representing general varieties.} Consider a finitary variety $\Var$ defined by a signature
$\Sigma = \Phi \cup \{\prog[skip],\prog[;]\}$ with arity map
$\ari:\Sigma\to \N$, and a set of equations $E$ that contains
the monoidal laws for $\{\prog[skip],\prog[;]\}$. Consider also
a monad $T:\cat\to\cat$ on a category with products. As the reader may have guessed, a Kleisli $T$-representation of a
$\Var$-object $A$ will be a morphism $\rho: A\to\EK(X)$. But in which
category? Since our starting point will always be a signature $\Sigma$, we define a Kleisli
$T$-representation of a $\Var$-object $A$ as a morphism
$\rho: A\to\EK(X)$ \emph{in the category of $\hspace{1pt}\Sigma$-algebras}, that is to
say a morphism $\rho$ which commutes with all the operations in the
signature but whose codomain may not live in $\Var$. The rationale for
this choice of category is the following: a group representation is not a group homomorphism because the whole point of a representation is to map group elements to important mathematical objects which do not form a group, for example real-valued matrices. Similarly, requiring the representation map to be a $\Var$-morphism would be way too stringent and would drastically limit the choice of possible semantics. By defining a Kleisli representation as a $\Sigma$-algebra morphism we do not require that
all equations in $E$ be \emph{valid} in $\EK(X)$, but we do require that
every operation in $\Sigma$ be \emph{interpretable} in $\EK(X)$. 

In order to endow $\EK(X)$ with a suitable algebraic structure from a
signature $\Sigma$, we proceed in the footsteps of
\cite{2001:PlotkinPowerAEAdequacy,2001:PlotkinPowerSemanticsAE,2003:PlotkinPowerEffects}:
for every $\sigma\in \Phi$ consider the set of natural transformations
$[\cat,\cat](T^{\ari(\sigma)},T)$ (as usual we take $T^0=\one$, the
constant functor on the final object 1). For each $\sigma\in \Phi$ we
choose an element
$\alpha^\sigma\in \left[\cat,\cat\right](T^{\ari(\sigma)},T)$ and
define the operation
$\lsem\sigma\rsem: \EK(X)^{\ari(\sigma)}\to \EK(X)$ by
\begin{equation}\label{eq:OperationDefinition}
  \lsem \sigma\rsem (a_1,\ldots,a_{\ari(\sigma)})=\alpha^\sigma_X\circ \langle a_1,\ldots,a_{\ari(\sigma)}\rangle
\end{equation}
We can now define a generic Kleisli representation as follows: a
  Kleisli $T$-representation of $A$ in $\cat_T$ is an assignment to
  every $\sigma\in\Phi$ of a natural transformation
  $\alpha^\sigma: T^{\ari(\sigma)}\to T$ together with a
  $\Sigma$-algebra morphism
\[
  \rho: (A,\prog[skip],\prog[\sComp],\sigma \in \Phi)
  \longrightarrow(\EK(X),\eta^T_X,\circ_T,
  (\lsem\sigma\rsem)_{\sigma\in\Phi})
\]


\paragraph{Naturality and abstraction} The requirement that operations be interpreted by natural transformations could be seen as either too strong or too weak. Too strong because it is a very restrictive condition whose justification is not immediately obvious. Too weak because it is strictly weaker than the requirement of \cite{2001:PlotkinPowerAEAdequacy,2001:PlotkinPowerSemanticsAE} defining algebraic operations where compatibility conditions with the strength and multiplication of the monad are assumed. So why have we chosen to focus on naturality?

First, algebraic operations are in general too restrictive for our
purpose. Examples of non-algebraic operations include
exception-handling operators \cite{2003:PlotkinPowerEffects} (which we
will address in Section~\ref{sec:combining}), `true' parallel
composition axiomatised by the \emph{exchange law}
\cite{1988:GischerEquational} and, more recently, examples from Game
Logic
\cite{2014:HansenPDL}. 
Actually, our work is to a large degree a systematic investigation of
what can be said about the semantics of programs with operations that
are \emph{not necessarily algebraic} in the sense of Plotkin and Power
-- a research direction already mentioned in
\cite{2001:PlotkinPowerSemanticsAE} -- for a wide range of monads.

Second, it is important to be able to consider \emph{sub-representations} and
\emph{quotient representations}, and, as we will show, naturality plays a key
role in allowing these to be defined.  
Often we need to abstract away details of a representation in a
large, fine-grained state space and build a representation in a
coarser one, but in such a way that both representations `agree' with
each other. Formally, for a quotient map $q : X \to Q$ and two Kleisli representations
$\rho : A \to \EK[T](X)$, $\rho' : A \to \EK[T](Q)$ we need that the equation
\begin{align*}
 T q \comp \rho(a) = \rho'(a) \comp q 
\end{align*}
holds for all programs $a \in A$.  Abstracting and
then interpreting should be the same as interpreting and then
abstracting. The naturality of program operations allows a \emph{compositional} construction of quotient representations
because it allows to prove that the equation
$T q \comp \rho(a) = \rho'(a) \comp q$ holds just by showing that it
holds for atomic programs. For sequential composition, this follows from the naturality of $\mu$: 
\begin{equation*}\label{diag:abstraction}
  \xymatrix@R=4ex {
    X\ar[r]^{\rho(a)}\ar[d]_{q} & TX\ar[r]^{T\rho(b)}\ar[d]^{Tq} & T^2X\ar[r]^{\mu_X}\ar[d]^{T^2q} & TX\ar[d]^{Tq}\\
    Q\ar[r]_{\rho'(a)} & TQ\ar[r]_{T\rho'(b)} & T^2Q\ar[r]_{\mu_Q} &
    TQ }
\end{equation*}
and for other operations the naturality requirement on $\alpha : T \times T \to T$
makes the following diagram also commute.
\[
  \xymatrix@R=4ex@C=12ex {
    X\ar[r]^-{\langle \rho(a),  \rho(b)\rangle}\ar[d]_{q} & TX\times TX\ar[r]^{\alpha_X}\ar[d]^{Tq \times  Tq} & TX\ar[d]^{Tq}\\
    Q\ar[r]_-{\langle \rho'(a), \rho'(b)\rangle} & TQ\times
    TQ\ar[r]_{\alpha_Q} & TQ }
\]

\noindent
So naturality allows to freely extend an abstraction from atomic
programs to all programs in the language.

\vspace{-1ex}
\section{Interpreting constants and operations}\label{sec:operations}

\subsection{Constants}
The following result, despite its simplicity, gives a very
general characterisation of natural transformations
$\underline{1}\to T$.

\begin{theorem}
  \label{thm:TNfromInitCriterion}
  Let $\cat$ be a category with an initial object $\emptyset$ and an
  object $1$ such that $ \cat(1, \_) \cong \Id$, then we have the sequence of bijections below.
  \[
    \left[\cat,\cat\right](\underline{1},F)\cong
    \cat(1,F\emptyset)\cong F\emptyset
  \]
\end{theorem}
\noindent Note that in $\Set$ the result above is a trivial consequence of the
Yoneda lemma, since $\underline{1}$ is representable as
$\hom(\emptyset,-)$.  

\begin{example}
\begin{enumerate}
\item The $\maybe$ monad $\Maybe$ has exactly one natural transformation
  $\left[\Set,\Set\right](\underline{1},\Maybe)=\Maybe \emptyset = 
  \{\lambda x. \ast\}$.
\item
  $\left[\Set,\Set\right](\underline{1},\Dist)=\Dist\emptyset=\emptyset$,
  and since the category $\Pol$ of Polish spaces satisfies the
  assumptions of Theorem \ref{thm:TNfromInitCriterion}, it is also the
  case that $\left[\Pol,\Pol\right](\underline{1},\Giry)=\emptyset$. Probabilistic programs do not support partial computations.
\item
  $\left[\Set,\Set\right](\underline{1},\HybM)=\HybM\emptyset=\emptyset$
  and thus hybrid programs also do not support an interpretation
  of failure.
\end{enumerate}
\end{example}

\subsection{Operations}
\paragraph{4.2.1. Coproducts of $\hom$ functors} We start with those
functors that can be written as coproducts of $\hom$ functors, since
they can be treated completely straightforwardly.  Actually, using the
notions of container and fibration, \cite{abbott03} already provides a
powerful representation theorem for natural transformations
$T \times T \to T$ when $T$ is one such functor.  In order to keep
this paper self-contained, however, we introduce a direct, equivalent
result that does not need the notion of fibration nor the notion of
container.

\begin{theorem}\label{thm:PlusCoprodHom} 
  If $F: \Set\to\Set$ is a functor expressible as a coproduct of
  $\hom$ functors, i.e. if there exists a non-empty family
  $(X_i)_{i\in I}$ of sets such that $F=\coprod_{i\in I}\hom(X_i,-)$,
  then
  \[
    \left[\Set,\Set\right](F\times F,F)\cong \prod_{i,j\in I}
    F(X_i+X_j)
  \]
\end{theorem}

\begin{example}
  \label{examples_bin}
  The maybe monad $\Maybe$ can be written as the coproduct
  $\hom(\emptyset,-)+\hom(1,-)$. It follows from Theorem
  \ref{thm:PlusCoprodHom} that the possible interpretations of a
  binary operation $\Maybe^2 \to \Maybe$ are in bijective
  correspondence with the set
  $\Maybe(2)\times \Maybe(1)\times \Maybe(1)$, in particular there are
  exactly 12 natural transformations $\Maybe^2\to\Maybe$. The
  `$\Maybe(2)$-coordinate' specifies what a transformation does on pairs
  $(x,y)$ with $x,y\neq \ast$, viz. projecting to the left, to the
  right or mapping to $\ast$, the first $\Maybe(1)$-coordinate
  specifies what happens to pairs $(x,\ast), x\neq \ast$,
  viz. projecting to the left or the right, and similarly for the last
  coordinate and pairs $(\ast,y), y\neq \ast$. 
  
\end{example}

\begin{example}
  Since $\HybM = \coprod_{ d \in [0,\infty)} \hom([0,d], -)$,
  the natural transformations $\HybM \times \HybM \to \HybM$
  are in bijective correspondence with the set
  \begin{flalign*}
   \prod_{i,j\in[0,\infty)} \HybM([0,i]+[0,j]) 
 \end{flalign*}
 For an element $s$ of this set,  each $(i,j)$-coordinate
 $s_{ij}$ dictates what the transformation $\alpha^s$ does to pairs of
 evolutions with duration $[0,i]$ and $[0,j]$. In particular, it tells
 how the values in a given pair of evolutions $(f,g)$ of duration
 $[0,i]$ and $[0,j]$ are distributed in the new evolution: 
 one has the composition
    \[
      \xymatrix{ [0,k] \ar[r]^(0.40){s_{ij}}
        \ar@/_1.2pc/[rr]_{\alpha^s_X(f,g)} & [0,i] + [0,j]
        \ar[r]^(0.67){[f,g]} & X }
    \]
    which makes clear that for every element $a \in [0,k]$ the value
    $\alpha^s_X(f,g) \left (a \right )$ arises from one of the two
    starting functions ($f$ or $g$) and an element in their respective
    domain ($[0,i]$ or $[0,j]$). 
\end{example}


\noindent
To cover a broader spectrum of monads
we need to introduce some more sophisticated mathematics.

\paragraph{4.2.2. The presentation of $\Set$-valued functors.}

Let $\cat$ be a \emph{small} category and $F:\cat\to\Set$ be a
functor. We define the \emph{category of elements of $F$}, denoted
$\El[F]$, as the category whose objects are pairs $(C,\alpha)$ where
$C$ is an object in $\cat$ and $\alpha\in FC$. There exists a morphism
$\hat{f}:(C,\alpha)\to(D,\beta)$ in $\El[F]$ whenever there exists a
morphism $f:C\to D$ such that $Ff(\alpha)=\beta$. For every object
$(C,\alpha)$ in $\El[F]$, we will define the \emph{orbit} of
$(C,\alpha)$ as all the objects $(D,\beta)$ which can be reached from
$(C,\alpha)$ by a zigzag of morphisms in $\El[F]$. The decomposition
of $\El[F]$ in orbits is key to understanding natural transformations
involving $F$, since naturality is only a constraint on objects in the
same orbit. The category $\El[F]$ allows us to completely reconstruct $F$. Moreover, this reconstruction
process provides us with a presentation of $F$ as a colimit of
covariant $\hom$ functors.

\begin{theorem}[\cite{2012:MaclaneSheaves} I.5] \label{thm:FuncRep}
  Let $\Yoneda: \cat\op\to[\cat,\Set]$ denote the Yoneda embedding,
  and $\Forg_F:\El[F]\to\cat$ be the forgetful functor sending each
  pair $(A,\alpha)$ to the object $A$, then
  \[
    F\cong \colim
    \left(\El[F]\op\stackrel{\Forg_F\op}{\longrightarrow}\cat\op\stackrel{\Yoneda}{\longrightarrow}[\cat,\Set]\right)
  \]
\end{theorem}

\noindent
Theorem \ref{thm:FuncRep} gives us a way of presenting functors from a
\emph{small} category $\cat$ to $\Set$ as a colimit of $\hom$
functors, but what we really need are presentations of functors
$\Set\to\Set$. To move from $\cat\to\Set$ to $\Set\to\Set$ we need a
few relatively well-known definitions. Recall that an
object $A$ in a category $\cat$ is \emph{finitely presentable} if the
functor $\hom(A,-)$ preserves filtered colimits. In $\Set$ the
finitely presentable objects are precisely the finite sets. A category
$\cat$ is called \emph{locally finitely presentable} if it is
cocomplete and contains a small subcategory $\cat_\omega$ of finitely
presentable objects such that every object $A$ in $\cat$ is the
filtered colimit of the canonical diagram
$\Diag_A:\cat_\omega\downarrow A\to \cat$ sending each arrow of
$\cat_\omega\downarrow A$ to its domain. The category $\Set$ is
finitely presentable, since every set $X$ can written as
$\colim\Diag_X$ with $\Diag_X:\omega\downarrow X \to \Set$ and
$\omega$ the subcategory of $\Set$ consisting of elements $n\in\omega$
(this simply says that every set is the union of its finite subsets).
Finally, a functor is called \emph{finitary} if it preserves filtered
colimits. Finitary functors are
entirely determined by their restriction to $\cat_\omega$. In fact if
$F:\cat\to\cat$ is a finitary functor on a locally finitely
presentable category and $\Inc:\cat_\omega\inc\cat$ is the inclusion
functor, then $F$ can be written as the left Kan
extension $F= \Lan (F_f)$, where $F_f = F \comp \Inc$.
We refer the interested reader to the classic \cite{1994:AdamekLocally} for a full account of the theory of locally finitely presentable categories. 

\begin{proposition}\label{prop:FinFuncRep}
  Let $F:\Set\to\Set$ be finitary and let $\Inc:\omega\inc\Set$ be the
  inclusion functor, then
  \[
    F\cong \colim \left(\El[F_f]\op\stackrel{\Forg_{F_f}\op}{\longrightarrow}\omega\op
      \stackrel{\Inc\op}{\longrightarrow}\Set\op\stackrel{\Yoneda}{\longrightarrow}[\Set,\Set]\right)
  \]
\end{proposition}


\noindent The following result is a simple application of Proposition
\ref{prop:FinFuncRep} and the Yoneda lemma.

\begin{theorem}\label{thm:BinaryTNatFinitary}
  Let $F:\Set\to\Set$ be a finitary functor and let $F_f$ denote its
  restriction to $\omega$, then the set $[\Set,\Set](F\times F,F)$ is
  in one-to-one correspondence with the limit
  \begin{align}\label{eq:TNatProdLim}
    \lim \left(\El[F_f\times F_f] \stackrel{F_f \circ \Forg}{\longrightarrow}\Set\right)
  \end{align}
\end{theorem}

\noindent The hard work consists in computing the limit
(\ref{eq:TNatProdLim}) above.

\paragraph{4.2.3. A classification result for some multiset functors.}
Computing the limit (\ref{eq:TNatProdLim}) of Theorem
\ref{thm:BinaryTNatFinitary} for general multiset-type monads (see their definition in the
appendix) depends
heavily on the choice of semiring and may prove extremely
difficult. However, we do have an explicit characterisation in the
following useful case. We say that a semiring $S$ has \emph{the common
  integer divisor property} if for any $x,y\in S$ there exist an
\emph{invertible} element $r\in S$ and $m,n\in\N$ such that
\[
x=m. r:=\underbrace{r+\ldots+r}_{m\text{ times}}\hspace{5em}y=n .r
\]
We will refer to $r$ as a \emph{common integer divisor of $x$ and $y$}. Clearly the semiring $\N$ has this property since we can always pick $r=1$ (which is trivially invertible) and $m=x, n=y$. Similarly the semiring $\Q$ has this property: given two rationals $x=\frac{m_1}{n_1}$, $y=\frac{m_2}{n_2}$ we can choose $r=\frac{1}{n_1n_2}$ (which is invertible) and $m=m_1n_2, n=m_2n_1$. The semiring of real $\R$ does not have this property: if $\frac{x}{y}$ is irrational then there doesn't exist an $r\in \R$ with the desired property.

\begin{theorem}\label{thm:GMsetTNat}
Let $S$ be a semiring with the common  integer divisor property and let $\GMset[S]$ be the multiset monad for $S$, then $\left[\Set,\Set\right]((\GMset[S])^n,\GMset[S])$ is in
  one-to-one correspondence with the set of functions
  $\phi:S^n\to S^n$.
\end{theorem}


\paragraph{4.2.4. Some results for the Giry monad}\label{sec:Giry}
Due to its importance, we provide some detailed results for the Giry
monad $\Giry$, which clarifies what can be
expected of purely probabilistic program semantics. Before we turn to the full Giry monad, let us consider the rational
distribution monad $\Dist_r$. The following can be shown using the
same ideas as in the proof of Theorem \ref{thm:GMsetTNat}.
\begin{theorem}\label{thm:TNatRatDist}
The only natural transformations $\Dist_r\times\Dist_r\to \Dist_r$ are the
  convex combinations $+^\lambda, \lambda\in \unit\cap\Q$.
\end{theorem}

\noindent
Theorem \ref{thm:TNatRatDist} can be used to provide a full classification result for the Giry monad on $\Pol$. For this we use a set of criteria for functors $F, G: \Pol\to\Pol$ developed in \cite{2016:Machine} and \cite{2016:Concur} under which it can be shown that
\[
  \left[\Pol,\Pol\right](F,G)\cong \left[\Pol_f,\Pol_f\right](F_f,
  G_f)
\]
where $F_f$ is the restriction of $F$ to the category $\Pol_f$, the
category of finite Polish spaces (and similarly for $G_f$). These
criteria restrict both the domain and the codomain functors. We refer
the reader to \cite{2016:Machine} for more details; for our purpose it
will be enough to say that the Giry monad $\Giry$ always satisfies the
domain and codomain criteria (see \cite[Prop. 5.1]{2016:Machine}) and
that finite products of $\Giry$ satisfy the domain criteria (see
\cite[Prop. 17]{2016:Concur}). It follows that
\begin{equation}\label{eq:MachineGiry}
  \left[\Pol,\Pol\right](\Giry\times\Giry,\Giry)\cong \left[\Pol_f,\Pol_f\right](\Giry_f\times\Giry_f, \Giry_f)
\end{equation}
The isomorphism \eqref{eq:MachineGiry} allows the following
result to be established at the level of finite Polish sets, and then
lifted to the entire category $\Pol$.

\begin{theorem}\label{thm:PlusGiry}
  The only natural transformations $\Giry\times\Giry\to \Giry$ are the
  convex combinations $+^\lambda$ defined by the maps
  $+^\lambda_X: \Giry X\times \Giry X\to \Giry X,
  \lambda\in\left[0,1\right]$ such that
  \[
    (\mu +^\lambda_X \nu)(A)=\lambda\mu(A)+(1-\lambda)\nu(A)
  \]
  for any Borel subset $A$.
\end{theorem}

\noindent
The combination $\Giry \Maybe:\Pol\to\Pol$ (with $\Maybe$
topologised in the obvious way) is a monad -- by a
  straightforward generalisation of Theorem~\ref{theo_TM} -- which  is called the \emph{subdistribution monad}. It
behaves very similarly to $\Giry$, but allows failure since
$\Giry \Maybe\emptyset = 1 = \{\delta_\ast\}$, and thus by Corollary
\ref{thm:TNfromInitCriterion} there exists a natural transformation
$\underline{1}\to \Giry \Maybe$. Moreover, whilst $\El[\Giry]$ has
a single orbit, $\El[\Giry\Maybe]$ has a collection of orbits
labelled by $\lambda\in\unit$. This makes the set of natural
transformations $(\Giry\Maybe)^2\to\Giry\Maybe$ rather large, but
using the same technique as in the proofs of Theorems
\ref{thm:GMsetTNat} and \ref{thm:PlusGiry} we can describe them
concisely as follows.
\begin{theorem}\label{thm:plusSubDist}
 The set $\left[\Pol,\Pol\right](\Giry\Maybe\times \Giry\Maybe,\Giry\Maybe)$ is
 in one-to-one correspondence with continuous maps
 \[
   \phi: \unit\times \unit \to \{(r_1,r_2)\in\unit^2\mid r_1+r_2\leq 1\}
 \]
\end{theorem}

\paragraph{4.2.5. The case of the powerset monad.}
A functor is called $\kappa$-accessible, for a regular cardinal $\kappa$, if it preserves $\kappa$-filtered colimits. In particular $\aleph_0$-accessible is synonymous with finitary. Theorem \ref{thm:BinaryTNatFinitary} generalizes completely straightforwardly to $\kappa$-accessible functors. Using this, we can give a complete classification result for the set of natural transformations $(\Pow)^\lambda\to\Pow$ for the full powerset monad and any cardinal $\lambda$.

\begin{theorem}\label{thm:PlusPow}
  For a cardinal $\lambda$ the set of natural transformations
  $(\Pow)^\lambda\to\Pow$, where $(\Pow)^\lambda$ denotes the
  $\lambda$-fold product of $\Pow$ with itself, is in one-to-one
  correspondence with the set of non-increasing maps
  $2^\lambda\to 2^\lambda$ (for the product order on $2^\lambda$).
\end{theorem}


\section{Interpreting program axiomatics}\label{sec:axioms}
We explore whether it is possible for a given monad
$T$ to define binary operations via natural transformations $T\times T\to T$
is such a way that important axioms hold.

Let us first fix some terminology: consider a set $V$ of variables, a
signature $\Sigma$, and the corresponding set of terms
$\mathrm{Trm}_\Sigma(V)$. Now take two terms
$s,t \in \mathrm{Trm}_\Sigma(V)$. We say that $s = t$ is
\emph{satisfiable in Kleisli $T$-representations} if there exists an
assignment $\alpha^\sigma: T^{\ari(\sigma)}\to T$ for each operation
$\sigma\in\Sigma$ such that for every set $X$ and every map
$i : V \to \EK(X)$, $\lsem s\rsem_i =\lsem t\rsem_i$ holds for the
interpretation defined by Eq. \eqref{eq:OperationDefinition} in
$\EK(X)$. We will permit ourselves the slight abuse of language
consisting in saying that the choice of natural transformation itself
satisfies $s=t$, for example we will say that a natural transformation
$\alpha^\sigma: T^2\to T$ is commutative if it makes
$\sigma(x,y)=\sigma(y,x)$ satisfiable in Kleisli
$T$-representations. Finally, we will say that $s=t$ is \emph{valid in
  Kleisli $T$-representations} if it satisfied for \emph{any}
assignment $\alpha^\sigma: T^{\ari(\sigma)}\to T$.


\paragraph{5.1. Commutativity.} We start with the case of coproducts of $\hom$ functors.

\begin{proposition}\label{prop:commutativity}
  If $F=\coprod_{i\in I}\hom(X_i,-)$, then a natural transformation
  $\alpha: F^2\to F$ given by an element
  $(s_{ij})_{i,j\in I}\in\prod_{i,j\in I}F(X_i+X_j)$ via Thm
  \ref{thm:PlusCoprodHom} \emph{is commutative} iff for all
  $i,j\in I$, the equation below holds.
  \begin{flalign*}
    [i_2,i_1] \comp s_{ij} = s_{ji} : X_k \to X_j + X_i
  \end{flalign*}
  In particular, for all $i \in I$, $s_{ii}$ must be the map with empty
  domain.
\end{proposition}

\noindent
Proposition~\ref{prop:commutativity} tells us that commutative natural
transformations $T^2 \to T$ for coproducts of hom functors $T$ must
adhere to rather harsh conditions. This is illustrated by the
following examples.

\begin{example}
\begin{enumerate}
\item A natural transformation $\alpha: \Maybe^2\to \Maybe$ can only
  be commutative if its `$\Maybe(2)$ coordinate' (defined by Thm
  \ref{thm:PlusCoprodHom}) lies in the second summand of $\Maybe(2)$, i.e.\@ if two elements different than failure are mapped to
  failure.
\item There is no commutative natural transformation for hybrid
  programs because 
  there exists no real number $d \in [0,\infty)$ such that
  $[0,d] = \emptyset$.
\end{enumerate}
\end{example}


\noindent
For the other monads presented thus far we can look directly at the
classification results provided in Section \ref{sec:operations}. We
have for example:

\begin{enumerate}
\item The commutative natural transformations $\GMset[S]^2\to \GMset[S]$ are precisely given by the  maps $S^2\to S$, i.e. the transformations choosing an equally weighted sum of multisets for each orbit.
\item The unique commutative natural transformation $\Dist_r^2\to\Dist_r$ is the average transformation $+^{\frac{1}{2}}$, and by the same argument as in Section \ref{sec:Giry} this is also the case for the Giry monad $\Giry$.
\end{enumerate}

\paragraph{5.2. Idempotence.} For coproducts of $\hom$ functors we have:

\begin{proposition}\label{prop:idempotence}
  If $F=\coprod_{i\in I}\hom(X_i,-)$, then a natural transformation
  $\alpha: F^2\to F$ given by an element
  $(s_{ij})_{i,j\in I}\in\prod_{i,j\in I}F(X_i+X_j)$ via
  Thm. \ref{thm:PlusCoprodHom} \emph{is idempotent} iff for each
  $i\in I$, $s_{ii}\in F(X_i+X_i)$ is a map $s_{ii}: X_i\to X_i+X_i$
  such that  $\triangledown \comp s_{ii} = \id$ where
  $\triangledown : X_i + X_i \to X_i$ is the codiagonal map.
\end{proposition}


\begin{example} 
\begin{enumerate}
\item A natural transformation $\alpha: \Maybe^2\to \Maybe$
  is idempotent iff its `$\Maybe(2)$
  coordinate' is in the first summand of $\Maybe(2)$. In other words,
  if two elements different than failure are projected either to the
  left or to the right.
\item A natural transformation $\alpha:\HybM^2\to\HybM$ is idempotent
  iff for every two evolutions $(f,g)$ with the same duration $[0,d]$,
  $\alpha^s(f,g)$ has domain $[0,d]$, and for every element
  $a \in [0,d]$, $\alpha^s(f,g)(a) = f(a)$ or
  $\alpha^s(f,g)(a) = g(a)$.
\end{enumerate}
\end{example}

\noindent
It follows from Propositions \ref{prop:commutativity} and
\ref{prop:idempotence} that the combination of commutativity and
idempotence is \emph{never satisfiable for Kleisli
  $T$-representations} when $T$ is a non-constant coproduct of $\hom$
functors since there will then exist an $X_i\neq\emptyset$ in the
coproduct presentation of $T$ for which Proposition
\ref{prop:commutativity} requires that $s_{ii}$ be of type
$\emptyset\to X_i+X_i$ but for which Proposition
\ref{prop:idempotence} requires that $s_{ii}$ be of type
$X_i\to X_i+ X_i$. Thus no monad whose functor is a coproduct of
$\hom$ functors can support a programming language with a commutative
and idempotent binary operation on programs.

We summarise the characterisation of idempotent natural transformations for the other monads as follows:

\begin{enumerate}
\item The idempotent natural transformations $\GMset[\N]^2\to\GMset[\N]$ are those which are defined as projections on pairs of multisets of equal weights. Note that as in the case of coproducts of $\hom$ functors, the combination of commutativity and idempotence is never satisfiable in Kleisli $\GMset[\N]$-representations, since the commutative operations must take equally weighted sums, whereas the idempotent ones must put one weight to zero and the other to one on pairs of multisets of equal weights.
\item Theorem \ref{thm:TNatRatDist} guarantees that \emph{all} natural transformation $\Dist_r^2\to\Dist_r$ are idempotent, i.e. idempotence is valid over $\Dist_r$-representations, and similarly for $\Giry$. Note that maps $\Q^2\to \Q^2$ mapping each pair $(q_1,q_2)$ to a pair $(\lambda,1-\lambda)$ with $\lambda\in\Q\cap\unit$ also define the idempotent operations of $\GMset[\Q]$.
\end{enumerate}

\paragraph{5.3. Units} For coproducts of $\hom$ functors we have:

\begin{proposition}\label{prop:units}
  Consider a functor $F=\coprod_{i\in I}\hom(X_i,-)$, a natural
  transformation $\alpha^s: F^2\to F$, given by an element
  $(s_{ij})_{i,j\in I}\in\prod_{i,j\in I}F(X_i+X_j)$ via Thm
  \ref{thm:PlusCoprodHom}, and a natural transformation
  $u: \const{1} \to F$. The
  transformation $u: \const{1} \to F$ is a unit for $\alpha^s$ iff
  $\id = m \comp s_{ik} = n \comp s_{ki} : X_i \to X_i$ with
  $m : X_i + \emptyset \to X_i$, $n : \emptyset + X_i \to X_i$ the
  isomorphisms that define $\emptyset$ as the unit of coproducts,
  and $X_k = \emptyset$.
\end{proposition}


\begin{example}
  \begin{enumerate}
  \item A natural transformation $\alpha: \Maybe^2\to \Maybe$ has the
    unique natural transformation $\underline{1}\to \Maybe$ as unit
    iff its `$\Maybe(1)$ coordinates' are in the first summand of
    $\Maybe(1)$.
\item The unique natural transformation $\one\to\GMset[S]$ which pick the constant multiset with weight 0 defines a unit for any binary natural transformation $\GMset[S]^2\to \GMset[S]$ since $0$ is the unit of $S$. The unit axiom is thus \emph{valid} over $\GMset[S]$-representations.
\item Since there are no natural transformations $\one\to\Dist$, there can be no interpretation of constants in Kleisli $\Dist$-representations, and similarly for $\Dist_r$, $\Giry$ and $\HybM$.  
\item The sub-distribution monad $\Giry\Maybe$ offers both probabilistic
  choice and partial computation but the unit axiom will be problematic for any binary operation: consider the sub-distributions on the singleton set
  $1$, we need
\[
\alpha_1\left(\frac{1}{n}\delta_1 +\frac{n-1}{n}\delta_{\star},\delta_{\star}\right)=\frac{1}{n}\delta_1 +\frac{n-1}{n}\delta_{\star}
\]
and
\[
 \alpha_1\left(\delta_{\star},\frac{1}{n}\delta_1 +\frac{n-1}{n}\delta_{\star}\right)=\frac{1}{n}\delta_1 +\frac{n-1}{n}\delta_{\star}
\]
This means that the function $\phi$ defining $\alpha$ (Thm.~\ref{thm:plusSubDist}) must map $(\frac{n-1}{n},1)$ to $(1,0)$ and $(1,\frac{n-1}{n})$ to $(0,1)$ for each $n$, which clearly cannot be continuous.
\end{enumerate}
\end{example}


\paragraph{5.5. Absorption.} The absorption law states that for every program
$\prog[p] : X \to T X$ the equations below hold.
\begin{flalign*}
    \prog[p \sComp 0] = \prog[0] \hspace{0.3cm} \text{ (left absorption) } \hspace{0.5cm}
    \prog[0 \sComp p] = \prog[0] \hspace{0.3cm} \text{ (right absorption) } 
\end{flalign*}

\noindent
Even though they appear to be simple, to check their satisfiability in
general Kleisli representations is a surprisingly complex issue. On
the one hand, we have:
\begin{proposition}\label{prop:absorption}
  The axiom $\prog[0 \sComp p]  = \prog[0]$ holds in every Kleisli
  representation.
 \end{proposition}

\noindent
On the other hand, the axiom $\prog[p \sComp 0]= \prog[0]$ need not
hold, and this is exemplified by the partial hybrid paradigm
$\HybM \Maybe$ which is briefly introduced in the following section.

\begin{theorem}
  \label{theo_dist}
  Consider a monad $T : \Set \to \Set$ with a natural transformation
  $\prog[0] : \const{1}\to T$. If the condition $T \emptyset \cong 1$ holds
  then the axiom $\prog[p \sComp 0] = \prog[0]$ holds as well.
\end{theorem}

\section{Adding features by combining monads}\label{sec:combining}

\subsection{Failure and Tests}

There is a natural procedure for turning Kleisli representations
which cannot interpret failure -- such as those for the distribution
or hybrid monads -- into Kleisli representations which do support
failure. More interestingly though, this procedure also allows to
interpret tests
, and in the case of probabilistic programs one recovers
the interpretation of \texttt{if-then-else} commands first given by
Kozen 
\cite{K81c}. 

\paragraph{Adding failure.} Recall from Section \ref{sec:operations}
that neither probabilistic nor hybrid systems can handle failure. As
recently described in \cite{sokolova:calco2017} for
probabilistic systems, adding a distinguished state is an obvious
strategy to overcome this shortcoming. We thus consider composing
effect monads of interest with the $\maybe$ monad $\Maybe$. This can
always be done as the following result \cite{Ghani02} shows.
\begin{theorem}\label{theo_TM}
  For a monad $T : \Set \to \Set$ there exists a distributive
  law $\delta: \Maybe T \to T \Maybe$ defined by
  \begin{flalign*}
    \delta_X = [ T i_1, \eta^T_{X + 1} \comp i_2 ]
  \end{flalign*}
\end{theorem}
\noindent
In particular $T\Maybe$ always forms a monad of which $\Maybe$ is a
submonad. Combining monads in this way potentially addresses more than the problem of non-representability of
failure: since every $\Set$-monad $T$ is
strong we can always find a natural transformation
$\otimes^n: (-)^n\circ T\to T\circ (-)^n$ (we use the convention
$(-)^0=\underline{1}$ and $\otimes^0=\eta^T_1$). 
Hence, if a $\Set$-monad $S$ has a natural transformation
$\alpha: S^n\to S$, then $TS$ also has a natural transformation
of this type given by
\begin{equation}\label{eq:lifting}
\xymatrix@C=8ex
{
(TS)^n \ar[r]^{\otimes^n S}&T(S^n)\ar[r]^{T\alpha^n}& TS
}
\end{equation}
In other words, we can increase the list of available operations by composing monads. Note that the algebraic properties of $\alpha$ discussed in Section \ref{sec:axioms} are usually lost by the lifting described above. However if $T$ is commutative, then the lifting above will preserve associativity, commutativity and the existence of units. For $n=0$, \eqref{eq:lifting} shows that $T\Maybe$ always has a natural transformation 
\[
\xymatrix@C=8ex
{
\underline{1}\ar[r]^{\otimes^0=\eta^T_1}& T\underline{1} \ar[r]^{T\ast} &T\Maybe
}
\]
and we can thus always interpret failure in $T\Maybe$. In particular,
the monads $\Dist \Maybe$ and $\HybM \Maybe$ support an abort
operation $\prog[0]$, and, as shown in Section~\ref{sec:axioms}, the
axiom $\prog[0 \, ; p] = \prog[0]$ is valid in both paradigms. We also
have the bijection $\Dist \Maybe \emptyset \cong \Dist 1 \cong 1$ and
by an application of Theorem~\ref{theo_dist} this entails that the
axiom $\prog[p \, ; 0] = \prog[0]$ is valid in probabilistic
programming.  On the other hand, in the hybrid case we have the
bijections $\HybM \Maybe \emptyset \cong \HybM 1 \cong [0,\infty)$
which say that one can always interpret an abort operation that
produces failures for a given duration $[0,d]$. For each of these
operations the axiom $\prog[p \, ; 0] = \prog[0]$ cannot hold: on the
left side, the duration of an evolution produced by $\prog[p]$ is
added to that of $\prog[0]$. Hence, if the evolution of $\prog[p]$ is greater
than $0$ the sum is greater than the duration of the evolution
produced on the right side.

\paragraph{Adding Tests.} The combination $T \Maybe$ also provides
Kleisli $T\Maybe$-representations of useful fragments of Kleene
algebras with tests (KAT) \cite{kozen1997kleene}. Here we will only
consider the $\ast$-free fragment of KATs, i.e. idempotent semirings
with tests, or ISTs, which are defined as follows. An IST is a
two-sorted structure $(S,B,+,;,^-,0,1)$ such that $B\subseteq S$,
$(S,+,;,0,1)$ is an idempotent semiring, and $(B,+,;,^-,0,1)$ is a
boolean algebra.

Let us first deal with tests: each test $\tt b$ should be interpreted as a
predicate $\lsem {\tt b}\rsem: X\to 2$ on the state space, which is
equivalent to a map
\begin{equation}\label{eq:test}
X\to X+1, x\mapsto\begin{cases}x & \text{if }\lsem
 {\tt b}\rsem(x)=1\\
 \ast&\text{else}\end{cases} 
\end{equation}
We will therefore interpret tests in $\EK[T\Maybe](X)$ as the maps
$X\to T\Maybe X$ that can be built as
$\eta^T_{\Maybe X}\circ f:X\to \Maybe X\to T\Maybe X$, where $f$ is of
the shape \eqref{eq:test}. In particular, we automatically get a
two-sorted structure on $\EK[T\Maybe](X)$. 
 Note also that we can straightforwardly extend the
single-sorted definition of Kleisli representation in Section
\ref{sec:KlRep} to a two-sorted definition so that KAT-like
structures can be accomodated.

We now need to define a boolean algebra structure on tests in
$\EK[T\Maybe](X)$. The constants are as expected: $\lsem 1\rsem$ is the $T$-lifting of the unit of $\Maybe$, i.e.\@ the unit of $T\Maybe$, and $\lsem 0\rsem$ is the $T$-lifting of the unique natural transformation $\underline{1}\to\Maybe$.
The negation of a test $\overline{\tt b}$ is defined in
the obvious way and it follows that $\overline{1}=0$. Conjunction is simply Kleisli composition.
It is not hard to check by unravelling the definitions and using the
properties of the distributive law $\Maybe T\to T\Maybe$ that for any
two tests $\tt a,b$:
$\lsem \mathtt{b}\rsem\circ_{T\Maybe}\lsem \mathtt{a}\rsem=\lsem \mathtt{a \, ; \, b}\rsem$.

For disjunction, things are more complicated. We have seen in Theorem
\ref{thm:PlusCoprodHom} and Section \ref{sec:axioms}, that there is no
idempotent commutative binary natural transformation
$\Maybe\times\Maybe\to\Maybe$ with $\underline{1}\to\Maybe$ as
unit. In particular we cannot expect the disjunction $+$ to be
interpreted using a natural transformation
$T\Maybe\times T\Maybe\to T\Maybe$ of the shape \eqref{eq:lifting}. In
other words, simply composing $T$ with the $\maybe$ monad does not give
a non-deterministic choice in general.

However, there exists a natural and useful sub-class of programs in an
IST for which composing $T$ with the $\maybe$ monad does give
non-deterministic choice. We call this sub-class the
\texttt{if-then-else} fragment. It is given by:
\begin{align*}
\tt p = b \in Test \mid a \in At \mid skip \mid 0 \mid p \, ; \, p \mid p+_b p 
\end{align*}
where $\tt p+_b q$ reads `if $\tt b$ then $\tt p$ else $\tt q$'. This
restricted non-determinism has Kleisli $T\Maybe$-representations for
any monad $T$. To see why, choose any of the 3 natural
transformations $\alpha:\Maybe^2\to \Maybe$ (Prop.\ref{prop:units}) for which
$\underline{1}\to\Maybe$ is a unit, i.e. any of the transformations
mapping pairs $(x,\ast),(\ast,x)$ to $x$. With this choice made to resolve $\Maybe$-non-determinism, we define the Kleisli $T\Maybe$-representation of $\tt p+_b q$ in $\EK[T\Maybe](X)$ as:
\begin{equation}\label{eq:ifThenElse}
\lsem \mathtt{p+_b q}\rsem := T\alpha_X\circ \otimes_{\Maybe X}\circ \langle \lsem \mathtt{b;q}\rsem,\lsem\mathtt{\overline{b};q}\rsem\rangle
\end{equation}
When $\otimes$ is derived from the strength of $T$ we can show that
this Kleisli $T\Maybe$-representation of \texttt{if-then-else}
statements is \emph{sound} in the sense that expected algebraic
properties of the operators $\tt -+_b-$ are preserved by the
representation. It is routine by unravelling the definitions and using the explicit construction of strength in $\Set$ to show:
\begin{theorem}\label{thm:ifThenElse}
Let $\tt a,b$ be tests and consider a Kleisli $T\Maybe$-representation in $\EK[T\Maybe](X)$ interpreting tests according to \eqref{eq:test} and \texttt{if-then-else} according to \eqref{eq:ifThenElse}. Then\footnote{Note the similarity with the axioms of convex algebras.},
\begin{center}
\begin{tabular}{l l }
$\lsem \mathtt{p+_b p}\rsem=\lsem \mathtt{p}\rsem$ & $\lsem \mathtt{p+_b q}\rsem=\lsem \mathtt{q+_{~\overline{b}} p}\rsem$ \\
$\lsem \mathtt{p+_{1} q}\rsem=\lsem \mathtt{p}\rsem $ &
$\tt \lsem p+_a (q+_b r)\rsem=\lsem(p+_{a;b} (q +_{~\overline{a};b} r)\rsem$
\end{tabular}
\end{center}
\end{theorem}
\noindent In the case of probabilistic programs, the Kleisli $\Dist\Maybe$-represen-tation \eqref{eq:ifThenElse} of \texttt{if-then-else} is precisely the one given by Kozen in \cite{K81c}.

\subsection{Non-determinism}

In order to extend a given programming paradigm with non-deterministic
features, one may also wish to combine the underlying monad $T$ with the
powerset $\Pow$ since the following two natural
transformations always exist.
\begin{flalign*}
&\alpha : \Pow T \times \Pow T \to \Pow T, \hspace{0.5cm} \alpha_X (A,B) = A \cup B \\
& 0 : \underline{1} \to \Pow T, \hspace{1.6cm} 0_X (\ast) = \emptyset
\end{flalign*}
\noindent
These respect the unit, idempotence, commutativity, and associative laws for
$\prog[+]$. So, in principle, less well-behaved monads, like $\HybM$ and
$\Dist$, together with $\Pow$ could give rise to new monads that can
be used to represent richer programming languages.

\paragraph{A negative result for $\Pow\Dist$.}
As already shown in \cite{2003:VaraccaProbability}, there is no
distributive law $\Dist\Pow\to\Pow\Dist$. We now strengthen this
result to `there does not exist any monad structure on $\Pow \Dist$
whatsoever' which answers the question recently raised in
\cite{2017:KeimelPlotkin} about this problem. Following
Theorem~\ref{thm:TNfromInitCriterion} we have:

\begin{lemma}\label{lem:unitPowfDist}
  The only natural transformations $\eta: \Id\to \Pow\Dist$ are the
  constant transformation $\eta_X(x)=\emptyset$, and the natural
  transformation defined by $\eta_X(x)=\{\delta_x\}$.
\end{lemma}

\noindent
By combining the result above, the monadic laws
$\mu\circ\eta_T=\mu\circ T\eta=\Id$, and adapting the proof of
\cite{2003:VaraccaProbability} we can show that:
\begin{theorem}\label{thm:PlotkinOpenProb}
  There is no monad structure on $\Pow\Dist$.
\end{theorem}

\paragraph{A positive result for $\NPow \HybM$}
On a more positive note, we can combine the non-empty powerset monad
$\NPow$ with $\HybM$ using the following theorem.
\begin{theorem}
  \label{theo_QE}
  There exists a distributive law
  $\delta : \HybM \NPow \to \NPow \HybM$ defined by
    $\delta_X (f,d) = \left \{ (g,d) \in \HybM X \mid g \in f
    \right \}$
  where $g \in f$ is shorthand notation for the condition
  \begin{flalign*}
    \forall t \in [0,d] \> . \> g (t ) \in f (t)
  \end{flalign*}
\end{theorem}

\section{Conclusion}
We have developed a general definition of program semantics as \emph{Kleisli representation}, and studied whether monads describing various computational paradigms support Kleisli representations of classical programming constructs. We have shown that it is possible to explicitly and exhaustively answer this question by providing complete classification results of natural transformations of the type $(T)^n\to T$. In particular we have shown that monads for probabilistic and hybrid effects cannot possibly support the usual combination of failure and non-determinism, but the techniques developed are widely applicable. We have shown how combining effect monads with the $\maybe$ monad allows the representation of a failure, tests and \texttt{if-then-else} statements, arguably representing a sweet spot between between complexity and expressivity.

\bibliography{dahlqvist_neves}


\begin{thebibliography}{34}


\ifx \showCODEN    \undefined \def \showCODEN     #1{\unskip}     \fi
\ifx \showDOI      \undefined \def \showDOI       #1{#1}\fi
\ifx \showISBNx    \undefined \def \showISBNx     #1{\unskip}     \fi
\ifx \showISBNxiii \undefined \def \showISBNxiii  #1{\unskip}     \fi
\ifx \showISSN     \undefined \def \showISSN      #1{\unskip}     \fi
\ifx \showLCCN     \undefined \def \showLCCN      #1{\unskip}     \fi
\ifx \shownote     \undefined \def \shownote      #1{#1}          \fi
\ifx \showarticletitle \undefined \def \showarticletitle #1{#1}   \fi
\ifx \showURL      \undefined \def \showURL       {\relax}        \fi
\providecommand\bibfield[2]{#2}
\providecommand\bibinfo[2]{#2}
\providecommand\natexlab[1]{#1}
\providecommand\showeprint[2][]{arXiv:#2}

\bibitem[\protect\citeauthoryear{A. and H.}{A. and H.}{2017}]%
        {sokolova:calco2017}
\bibfield{author}{\bibinfo{person}{Sokolova A.} {and} \bibinfo{person}{Woracek
  H.}} \bibinfo{year}{2017}\natexlab{}.
\newblock \showarticletitle{Termination in Convex Sets of Distributions}. In
  \bibinfo{booktitle}{{\em CALCO}}. LIPIcs 72, \bibinfo{pages}{1--16}.
\newblock


\bibitem[\protect\citeauthoryear{Abbott, Altenkirch, and Ghani}{Abbott
  et~al\mbox{.}}{2003}]%
        {abbott03}
\bibfield{author}{\bibinfo{person}{M.~Gordon Abbott}, \bibinfo{person}{T.
  Altenkirch}, {and} \bibinfo{person}{N. Ghani}.}
  \bibinfo{year}{2003}\natexlab{}.
\newblock \showarticletitle{Categories of Containers}. In
  \bibinfo{booktitle}{{\em {FOSSACS}}} {\em (\bibinfo{series}{LNCS})},
  Vol.~\bibinfo{volume}{2620}. \bibinfo{publisher}{Springer},
  \bibinfo{pages}{23--38}.
\newblock


\bibitem[\protect\citeauthoryear{Ad{\'a}mek and Rosicky}{Ad{\'a}mek and
  Rosicky}{1994}]%
        {1994:AdamekLocally}
\bibfield{author}{\bibinfo{person}{J. Ad{\'a}mek} {and} \bibinfo{person}{J.
  Rosicky}.} \bibinfo{year}{1994}\natexlab{}.
\newblock \bibinfo{booktitle}{{\em Locally presentable and accessible
  categories}}. Vol.~\bibinfo{volume}{189}.
\newblock \bibinfo{publisher}{Cambridge University Press}.
\newblock


\bibitem[\protect\citeauthoryear{Dahlqvist, Danos, and Garnier}{Dahlqvist
  et~al\mbox{.}}{2016a}]%
        {2016:Machine}
\bibfield{author}{\bibinfo{person}{F. Dahlqvist}, \bibinfo{person}{V. Danos},
  {and} \bibinfo{person}{I. Garnier}.} \bibinfo{year}{2016}\natexlab{a}.
\newblock \showarticletitle{Giry and the Machine}. In \bibinfo{booktitle}{{\em
  {MFPS XXXII}}} {\em (\bibinfo{series}{ENTCS})}, Vol.~\bibinfo{volume}{325}.
  \bibinfo{pages}{85--110}.
\newblock


\bibitem[\protect\citeauthoryear{Dahlqvist, Danos, and Garnier}{Dahlqvist
  et~al\mbox{.}}{2016b}]%
        {2016:Concur}
\bibfield{author}{\bibinfo{person}{F. Dahlqvist}, \bibinfo{person}{V. Danos},
  {and} \bibinfo{person}{I. Garnier}.} \bibinfo{year}{2016}\natexlab{b}.
\newblock \showarticletitle{Robustly Parameterised Higher-Order Probabilistic
  Models}. In \bibinfo{booktitle}{{\em {CONCUR}}} {\em
  (\bibinfo{series}{LIPIcs})}, Vol.~\bibinfo{volume}{59}.
  \bibinfo{publisher}{Schloss Dagstuhl - Leibniz-Zentrum fuer Informatik},
  \bibinfo{pages}{23:1--23:15}.
\newblock


\bibitem[\protect\citeauthoryear{Foster, Kozen, Mamouras, Reitblatt, and
  Silva}{Foster et~al\mbox{.}}{2016}]%
        {2016:ProbNetKAT}
\bibfield{author}{\bibinfo{person}{N. Foster}, \bibinfo{person}{D. Kozen},
  \bibinfo{person}{K. Mamouras}, \bibinfo{person}{M. Reitblatt}, {and}
  \bibinfo{person}{A. Silva}.} \bibinfo{year}{2016}\natexlab{}.
\newblock \showarticletitle{Probabilistic {NetKAT}}. In
  \bibinfo{booktitle}{{\em ESOP 2016}} {\em (\bibinfo{series}{LNCS})}.
  Springer, \bibinfo{pages}{282--309}.
\newblock


\bibitem[\protect\citeauthoryear{Giry}{Giry}{1982}]%
        {1981:Giry}
\bibfield{author}{\bibinfo{person}{M. Giry}.} \bibinfo{year}{1982}\natexlab{}.
\newblock \showarticletitle{A categorical approach to probability theory}.
\newblock In \bibinfo{booktitle}{{\em Categorical aspects of topology and
  analysis}}. \bibinfo{publisher}{Springer}, \bibinfo{pages}{68--85}.
\newblock


\bibitem[\protect\citeauthoryear{Gischer}{Gischer}{1988}]%
        {1988:GischerEquational}
\bibfield{author}{\bibinfo{person}{J.~L. Gischer}.}
  \bibinfo{year}{1988}\natexlab{}.
\newblock \showarticletitle{The equational theory of pomsets}.
\newblock \bibinfo{journal}{{\em Theoretical Computer Science\/}}
  \bibinfo{volume}{61}, \bibinfo{number}{2-3} (\bibinfo{year}{1988}),
  \bibinfo{pages}{199--224}.
\newblock


\bibitem[\protect\citeauthoryear{Goodman, Mansinghka, Roy, Bonawitz, and
  Tenenbaum}{Goodman et~al\mbox{.}}{2012}]%
        {goodman2012church}
\bibfield{author}{\bibinfo{person}{N. Goodman}, \bibinfo{person}{V.
  Mansinghka}, \bibinfo{person}{D.~M Roy}, \bibinfo{person}{K. Bonawitz}, {and}
  \bibinfo{person}{J. Tenenbaum}.} \bibinfo{year}{2012}\natexlab{}.
\newblock \showarticletitle{Church: a language for generative models}.
\newblock \bibinfo{journal}{{\em arXiv preprint arXiv:1206.3255\/}}
  (\bibinfo{year}{2012}).
\newblock


\bibitem[\protect\citeauthoryear{Hansen, Kupke, and Leal}{Hansen
  et~al\mbox{.}}{2014}]%
        {2014:HansenPDL}
\bibfield{author}{\bibinfo{person}{H.~H. Hansen}, \bibinfo{person}{C. Kupke},
  {and} \bibinfo{person}{R.~A. Leal}.} \bibinfo{year}{2014}\natexlab{}.
\newblock \showarticletitle{Strong Completeness for Iteration-Free Coalgebraic
  Dynamic Logics}. In \bibinfo{booktitle}{{\em {TCS} 2014, Proceedings}} {\em
  (\bibinfo{series}{LNCS})}, Vol.~\bibinfo{volume}{8705}.
  \bibinfo{publisher}{Springer}, \bibinfo{pages}{281--295}.
\newblock


\bibitem[\protect\citeauthoryear{Haywood}{Haywood}{2011}]%
        {haywood2011symmetries}
\bibfield{author}{\bibinfo{person}{S. Haywood}.}
  \bibinfo{year}{2011}\natexlab{}.
\newblock \bibinfo{booktitle}{{\em Symmetries and conservation laws in particle
  physics: an introduction to group theory for particle physicists}}.
\newblock \bibinfo{publisher}{World scientific}.
\newblock


\bibitem[\protect\citeauthoryear{H{\"o}fner}{H{\"o}fner}{2009}]%
        {hofner09}
\bibfield{author}{\bibinfo{person}{P. H{\"o}fner}.}
  \bibinfo{year}{2009}\natexlab{}.
\newblock {\em \bibinfo{title}{Algebraic calculi for hybrid systems}}.
\newblock \bibinfo{thesistype}{Ph.D. Dissertation}. \bibinfo{school}{University
  of Augsburg}.
\newblock


\bibitem[\protect\citeauthoryear{Keimel and Plotkin}{Keimel and
  Plotkin}{2017}]%
        {2017:KeimelPlotkin}
\bibfield{author}{\bibinfo{person}{K. Keimel} {and} \bibinfo{person}{G.
  Plotkin}.} \bibinfo{year}{2017}\natexlab{}.
\newblock \showarticletitle{Mixed Powerdomains for probability and non
  determinism}.
\newblock \bibinfo{journal}{{\em To Appear in LMCS\/}} (\bibinfo{year}{2017}).
\newblock


\bibitem[\protect\citeauthoryear{Kozen}{Kozen}{1981}]%
        {K81c}
\bibfield{author}{\bibinfo{person}{D. Kozen}.} \bibinfo{year}{1981}\natexlab{}.
\newblock \showarticletitle{Semantics of probabilistic programs}.
\newblock \bibinfo{journal}{{\em J. Comput. Syst. Sci.\/}}
  \bibinfo{volume}{22}, \bibinfo{number}{3} (\bibinfo{date}{June}
  \bibinfo{year}{1981}), \bibinfo{pages}{328--350}.
\newblock


\bibitem[\protect\citeauthoryear{Kozen}{Kozen}{1997}]%
        {kozen1997kleene}
\bibfield{author}{\bibinfo{person}{D. Kozen}.} \bibinfo{year}{1997}\natexlab{}.
\newblock \showarticletitle{Kleene algebra with tests}.
\newblock \bibinfo{journal}{{\em ACM Transactions on Programming Languages and
  Systems (TOPLAS)\/}} \bibinfo{volume}{19}, \bibinfo{number}{3}
  (\bibinfo{year}{1997}), \bibinfo{pages}{427--443}.
\newblock


\bibitem[\protect\citeauthoryear{L{\"{u}}th and Ghani}{L{\"{u}}th and
  Ghani}{2002}]%
        {Ghani02}
\bibfield{author}{\bibinfo{person}{C. L{\"{u}}th} {and} \bibinfo{person}{N.
  Ghani}.} \bibinfo{year}{2002}\natexlab{}.
\newblock \showarticletitle{Composing monads using coproducts}. In
  \bibinfo{booktitle}{{\em {ICFP}}},
  \bibfield{editor}{\bibinfo{person}{M.~Wand} {and}
  \bibinfo{person}{S.~L.~Peyton Jones}} (Eds.). \bibinfo{publisher}{{ACM}},
  \bibinfo{pages}{133--144}.
\newblock


\bibitem[\protect\citeauthoryear{MacLane and Moerdijk}{MacLane and
  Moerdijk}{2012}]%
        {2012:MaclaneSheaves}
\bibfield{author}{\bibinfo{person}{S. MacLane} {and} \bibinfo{person}{I.
  Moerdijk}.} \bibinfo{year}{2012}\natexlab{}.
\newblock \bibinfo{booktitle}{{\em Sheaves in geometry and logic: A first
  introduction to topos theory}}.
\newblock \bibinfo{publisher}{Springer}.
\newblock


\bibitem[\protect\citeauthoryear{Moggi}{Moggi}{1989}]%
        {1989:MoggiComputational}
\bibfield{author}{\bibinfo{person}{E. Moggi}.} \bibinfo{year}{1989}\natexlab{}.
\newblock \showarticletitle{Computational lambda-calculus and monads}. In
  \bibinfo{booktitle}{{\em {LICS}}}. IEEE, \bibinfo{pages}{14--23}.
\newblock


\bibitem[\protect\citeauthoryear{Moggi}{Moggi}{1991}]%
        {1991:MoggiNotions}
\bibfield{author}{\bibinfo{person}{E. Moggi}.} \bibinfo{year}{1991}\natexlab{}.
\newblock \showarticletitle{Notions of computation and monads}.
\newblock \bibinfo{journal}{{\em Information and computation\/}}
  \bibinfo{volume}{93}, \bibinfo{number}{1} (\bibinfo{year}{1991}),
  \bibinfo{pages}{55--92}.
\newblock


\bibitem[\protect\citeauthoryear{Neves, Barbosa, Hofmann, and Martins}{Neves
  et~al\mbox{.}}{2016}]%
        {neves15}
\bibfield{author}{\bibinfo{person}{R. Neves}, \bibinfo{person}{L.~S. Barbosa},
  \bibinfo{person}{D. Hofmann}, {and} \bibinfo{person}{M.~A. Martins}.}
  \bibinfo{year}{2016}\natexlab{}.
\newblock \showarticletitle{Continuity as a computational effect}.
\newblock \bibinfo{journal}{{\em JLAMP\/}} (\bibinfo{year}{2016}).
\newblock


\bibitem[\protect\citeauthoryear{P. and F.}{P. and F.}{2014}]%
        {pmlr-v32-paige14}
\bibfield{author}{\bibinfo{person}{Brooks P.} {and} \bibinfo{person}{Wood F.}}
  \bibinfo{year}{2014}\natexlab{}.
\newblock \showarticletitle{A Compilation Target for Probabilistic Programming
  Languages}. In \bibinfo{booktitle}{{\em International Conference on Machine
  Learning}} {\em (\bibinfo{series}{Proceedings of Machine Learning
  Research})}, Vol.~\bibinfo{volume}{32}. \bibinfo{pages}{1935--1943}.
\newblock


\bibitem[\protect\citeauthoryear{Perko}{Perko}{2013}]%
        {perko2013differential}
\bibfield{author}{\bibinfo{person}{Lawrence Perko}.}
  \bibinfo{year}{2013}\natexlab{}.
\newblock \bibinfo{booktitle}{{\em Differential equations and dynamical
  systems}}. Vol.~\bibinfo{volume}{7}.
\newblock \bibinfo{publisher}{Springer Science \& Business Media}.
\newblock


\bibitem[\protect\citeauthoryear{Platzer}{Platzer}{2010}]%
        {platzer10}
\bibfield{author}{\bibinfo{person}{A. Platzer}.}
  \bibinfo{year}{2010}\natexlab{}.
\newblock \bibinfo{booktitle}{{\em Logical Analysis of Hybrid Systems: Proving
  Theorems for Complex Dynamics}}.
\newblock \bibinfo{publisher}{Springer}, \bibinfo{address}{Heidelberg}.
\newblock


\bibitem[\protect\citeauthoryear{Plotkin and Power}{Plotkin and Power}{2001a}]%
        {2001:PlotkinPowerAEAdequacy}
\bibfield{author}{\bibinfo{person}{G. Plotkin} {and} \bibinfo{person}{J.
  Power}.} \bibinfo{year}{2001}\natexlab{a}.
\newblock \showarticletitle{Adequacy for algebraic effects}. In
  \bibinfo{booktitle}{{\em International Conference on Foundations of Software
  Science and Computation Structures ({FoSSACS})}}.
  \bibinfo{publisher}{Springer}, \bibinfo{pages}{1--24}.
\newblock


\bibitem[\protect\citeauthoryear{Plotkin and Power}{Plotkin and Power}{2001b}]%
        {2001:PlotkinPowerSemanticsAE}
\bibfield{author}{\bibinfo{person}{G. Plotkin} {and} \bibinfo{person}{J.
  Power}.} \bibinfo{year}{2001}\natexlab{b}.
\newblock \showarticletitle{Semantics for algebraic operations}.
\newblock \bibinfo{journal}{{\em ENTCS\/}}  \bibinfo{volume}{45}
  (\bibinfo{year}{2001}), \bibinfo{pages}{332--345}.
\newblock


\bibitem[\protect\citeauthoryear{Plotkin and Power}{Plotkin and Power}{2003}]%
        {2003:PlotkinPowerEffects}
\bibfield{author}{\bibinfo{person}{G. Plotkin} {and} \bibinfo{person}{J.
  Power}.} \bibinfo{year}{2003}\natexlab{}.
\newblock \showarticletitle{Algebraic operations and generic effects}.
\newblock \bibinfo{journal}{{\em Applied Categorical Structures\/}}
  \bibinfo{volume}{11}, \bibinfo{number}{1} (\bibinfo{year}{2003}),
  \bibinfo{pages}{69--94}.
\newblock


\bibitem[\protect\citeauthoryear{Power}{Power}{1999}]%
        {1999:PowerEnriched}
\bibfield{author}{\bibinfo{person}{J. Power}.} \bibinfo{year}{1999}\natexlab{}.
\newblock \showarticletitle{Enriched {L}awvere theories}.
\newblock \bibinfo{journal}{{\em Theory and Applications of Categories\/}}
  \bibinfo{volume}{6}, \bibinfo{number}{7} (\bibinfo{year}{1999}),
  \bibinfo{pages}{83--93}.
\newblock


\bibitem[\protect\citeauthoryear{Smolka, Kumar, Foster, Kozen, and
  Silva}{Smolka et~al\mbox{.}}{2017}]%
        {2017:CantorScott}
\bibfield{author}{\bibinfo{person}{S. Smolka}, \bibinfo{person}{P. Kumar},
  \bibinfo{person}{N. Foster}, \bibinfo{person}{D. Kozen}, {and}
  \bibinfo{person}{A. Silva}.} \bibinfo{year}{2017}\natexlab{}.
\newblock \showarticletitle{Cantor meets {S}cott}. In \bibinfo{booktitle}{{\em
  POPL 2017}}. ACM, \bibinfo{pages}{557--571}.
\newblock


\bibitem[\protect\citeauthoryear{Sokolova and Woracek}{Sokolova and
  Woracek}{2017}]%
        {2017:Sokolova}
\bibfield{author}{\bibinfo{person}{A. Sokolova} {and} \bibinfo{person}{H.
  Woracek}.} \bibinfo{year}{2017}\natexlab{}.
\newblock \showarticletitle{Termination in Convex Sets of Distributions}. In
  \bibinfo{booktitle}{{\em CALCO 2017}}.
\newblock


\bibitem[\protect\citeauthoryear{Steinberg}{Steinberg}{2011}]%
        {steinberg2011representation}
\bibfield{author}{\bibinfo{person}{B. Steinberg}.}
  \bibinfo{year}{2011}\natexlab{}.
\newblock \bibinfo{booktitle}{{\em Representation theory of finite groups: an
  introductory approach}}.
\newblock \bibinfo{publisher}{Springer}.
\newblock


\bibitem[\protect\citeauthoryear{Suenaga and Hasuo}{Suenaga and Hasuo}{2011}]%
        {suenaga11}
\bibfield{author}{\bibinfo{person}{Kohei Suenaga} {and} \bibinfo{person}{Ichiro
  Hasuo}.} \bibinfo{year}{2011}\natexlab{}.
\newblock \showarticletitle{Programming with infinitesimals: A while-language
  for hybrid system modeling}.
\newblock \bibinfo{journal}{{\em Automata, Languages and Programming\/}}
  (\bibinfo{year}{2011}), \bibinfo{pages}{392--403}.
\newblock


\bibitem[\protect\citeauthoryear{Tabuada}{Tabuada}{2009}]%
        {tabuada09}
\bibfield{author}{\bibinfo{person}{P. Tabuada}.}
  \bibinfo{year}{2009}\natexlab{}.
\newblock \bibinfo{booktitle}{{\em Verification and Control of Hybrid Systems -
  {A} Symbolic Approach}}.
\newblock \bibinfo{publisher}{Springer}.
\newblock


\bibitem[\protect\citeauthoryear{Varacca}{Varacca}{2003}]%
        {2003:VaraccaProbability}
\bibfield{author}{\bibinfo{person}{D. Varacca}.}
  \bibinfo{year}{2003}\natexlab{}.
\newblock {\em \bibinfo{title}{Probability, nondeterminism and concurrency: two
  denotational models for probabilistic computation}}.
\newblock \bibinfo{thesistype}{Ph.D. Dissertation}. \bibinfo{school}{BRICS}.
\newblock


\bibitem[\protect\citeauthoryear{Wood, van~de Meent, and Mansinghka}{Wood
  et~al\mbox{.}}{2014}]%
        {wood-aistats-2014}
\bibfield{author}{\bibinfo{person}{F. Wood}, \bibinfo{person}{J.~W. van~de
  Meent}, {and} \bibinfo{person}{V. Mansinghka}.}
  \bibinfo{year}{2014}\natexlab{}.
\newblock \showarticletitle{A New Approach to Probabilistic Programming
  Inference}. In \bibinfo{booktitle}{{\em International conference on
  Artificial Intelligence and Statistics}}. \bibinfo{pages}{1024--1032}.
\newblock


\end{thebibliography}

\appendix

\section{Some programming languages generated by the hybrid monad}

The programming language $\prog[Hyb](X)$, analysed in Sections
\ref{sect:illustr} and \ref{sec:KlRep}, is \emph{time-triggered}: a
program $(\dot{x}_1 = t_1, \dots, \dot{x}_n = t_n \> \& \> d)$
terminates precisely when the instant of time $d$ is achieved. It is
also possible to consider \emph{event-triggered} languages by forcing
a program to terminate \emph{as soon as} a certain event occurs.  For
this, however, one needs to be very strict on the kinds of event
allowed, otherwise there might not exist an earliest time at which an
event happens.

Consider a finite set of real-valued variables
$X = \{ x_1, \dots, x_n\}$ and denote by $\At(X)$ the set given by the
grammar
\begin{align*}
  & \varphi = (x_1 := t, \dots, x_n := t) \mid
    (\dot{x}_1 = t, \dots, \dot{x}_n = t \> \& \> \psi),  \\
  &
  t = r \mid r \cdot x \mid t + t, \\
  & \psi = t \leq t \mid t \geq t \mid
    \psi \wedge \psi \mid \psi \vee \psi
\end{align*}
where $x \in X$.  Then, for a predicate $\psi$ define
$\lsem {\psi} \rsem \subseteq \R^n$ by,
\begin{flalign*}
  \lsem {t_1 \leq t_2} \rsem & = \left \{ (v_1,\dots,v_n) \in \R^n
    \mid \lsem {t_1} \rsem \leq \lsem {t_2} \rsem \right \} \\
  \lsem {t_1 \geq t_2} \rsem & = \left \{ (v_1,\dots,v_n) \in \R^n
    \mid
    \lsem {t_1} \rsem \geq \lsem {t_2} \rsem \right \} \\
  \lsem {\psi_1 \wedge \psi_1} \rsem & =  \lsem {\psi_1} \rsem \cap \lsem {\psi_2} \rsem \\
  \lsem {\psi_1 \vee \psi_1} \rsem & = \lsem {\psi_1} \rsem \cup \lsem {\psi_2}
  \rsem
\end{flalign*}

\begin{proposition}
  \label{cor:closed}
  For every predicate $\psi$ the set $\lsem {\psi} \rsem$ is closed in
  $\R^n$.
\end{proposition}

\begin{proof}
  Start with $\lsem {t_1 \leq t_2} \rsem$. We will to show that for every
  family of real numbers $(a_i,b_i)_{i \in n}$ and
  $(c,d) \in \R \times \R$ the set
  \begin{flalign*}
    A = \left \{ (v_1,\dots,v_n) \in \R^n \mid \sum\nolimits_{i
        \in I} a_i v_i + c \leq \sum\nolimits_{i \in I} b_i v_i + d
    \right \}
  \end{flalign*}
  is closed in $\R^n$. Recall that the order relation
  $R_{\leq} \subseteq \R \times \R$ is closed in the Euclidean space
  $\R \times \R$ and consider two
  maps $f,g : \R^n \to \R$ defined by
  \begin{align*}
    f(v_1,\dots,v_n) & = \sum\nolimits_{i \in I} a_i v_i + c,
    \\ g(v_1,\dots,v_n) & = \sum\nolimits_{i \in I} b_i v_i +
    d
  \end{align*}
  Clearly both $f$ and $g$ are continuous since they can be written as
  compositions of multiplication and addition maps. Moreover,
  $A = \langle {f, g} \rangle^{-1}(R_\leq)$ and therefore $A$ is
  closed in $\R^n$.

  An analogous reasoning applies to $\lsem {t_1 \geq t_2} \rsem$ since
  the order relation $R_{\geq} \subseteq \R \times \R$ is also closed.
  Finally, for the cases that involve conjunction and disjunction
  apply the property of closed sets being closed under intersections
  and finite unions.
\end{proof}

\begin{corollary}
  \label{cor:min}
  Let us consider a program
  $(\dot{x}_1 = t_1, \dots, \dot{x}_n = t_n \> \& \> \psi)$, its
  solution $\phi : \R^n \times [0,\infty) \to \R^n$, and a valuation
  $(v_1,\dots,v_n) \in \R^n$. If there exists a time instant
  $t \in [0,\infty)$ such that
  $\phi(v_1,\dots,v_n,t) \in \lsem {\psi} \rsem$ then there exists a
  \emph{smallest} time instant that also satisfies this condition.
\end{corollary}

\begin{proof}
  Using Proposition~\ref{cor:closed} we can easily show that
  the set
  $\phi(v_1,\dots,v_n, -)^{-1} (\lsem{\psi}\rsem) \cap [0,t]$ is
  compact, and consequently that  it has a minimum.
\end{proof}

\noindent
We can now introduce the following event-triggered programming language. Define
the interpretation map
\begin{flalign*}
  \At(X) \to \EK[\HybM](\R^n)
\end{flalign*}
as the one that sends $(x_1 := t_1, \dots, x_n := t_n)$ to the
function $\R^n \to \HybM(\R^n)$ defined by
\begin{flalign*}
  (v_1,\dots,v_n) \mapsto \eta \left (\lsem {t_1} \rsem ,\dots ,\lsem {t_n} \rsem \right)
\end{flalign*}
and $(\dot{x}_1 = t_1, \dots, \dot{x}_n = t_n \> \& \> \psi)$ to the
function $\R^n \to \HybM(\R^n)$ defined by
\begin{flalign*}
  (v_1,\dots,v_n) \mapsto (\phi(v_1,\dots,v_n, -), d)
\end{flalign*}
\noindent
where $d$ is the smallest time instant that intersects
$\lsem{\psi}\rsem$ if
$\phi^{-1}(v_1,\dots,v_n,-) \cap \lsem{\psi} \rsem \not = \emptyset$
(Corollary~\ref{cor:min}) and $0$ otherwise.  The free monoid
extension of this interpretation map induces a programming language
  \[
    \prog[p]=\prog[a\in At(X) \mid skip\mid p\sComp p]
  \]
whose semantics is a Kleisli representation.

\begin{example}[Bouncing ball]
Consider a bouncing ball dropped at a positive height $p$ and with
no initial velocity $v$. Due to the gravitational acceleration $g$,
it falls into the ground and then bounces back up, losing part of
its kinetic energy in the process. Consider the program,
\begin{flalign*}
  (\dot{p} = v, \dot{v} = g \> \& \> p \leq 0 \wedge v \leq 0) \prog[\sComp]  (v := v \times -0.5)
\end{flalign*}
which is here denoted by $\prog[b]$. The composition
\begin{flalign*}
 (v : = 0, p := 5) \prog[\sComp] \prog[b] \prog[\sComp] \prog[b] \prog[\sComp] \prog[b]
\end{flalign*}
encodes the action of dropping the ball at the
height of five meters and letting it bounce exactly three times.
The projection on $p$ of this program yields the plot below.
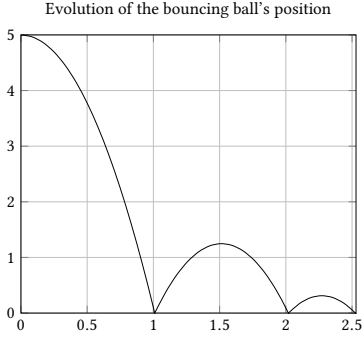
\begin{center}
  \scalebox{0.65}{
    \begin{tikzpicture}
      \begin{axis}[ title={Evolution of the bouncing ball's
          position}, ymin= 0, ymax = 5, xmin = 0, xmax = 2.53,
        grid = major ]
        \addplot [smooth, domain = 0:2] { 5 - (1/2)*9.8*x^2 };
        \addplot[smooth, shift = {(101.0,0.0)}, domain = 0:2] {
          4.945*x - (1/2)*9.8*x^2 }; \addplot[smooth, shift =
        {(202.0,0.0)}, domain = 0:2] { 2.473*x - (1/2)*9.8*x^2 };
      \end{axis}
    \end{tikzpicture}
  }
\end{center}
\end{example}

\noindent
Our next hybrid programming language is closely related to
H\"{o}fner's algebra of hybrid systems \cite{hofner09}. Recall the
interpretation map $\At(X) \to \EK[\HybM](\R^n)$ for the
event-triggered programming language and compose it with the function
$\EK[\HybM](\R^n) \to \EK[\HybM](\R^n \times 2)$ that sends $f$ to the
map $g$ defined by
\begin{align*}
  g(v,\bot) & = \eta(v,\bot) \\
  g(v,\top) & = (h,d)            
\end{align*}
where $f(x) = (f(x, -), d)$ and $h(t) = (f(x,t),\bot)$ if $t \not = d$
and $h(t) = (f(x,t),\top)$ otherwise. The free monoid extension
of the composition
\begin{flalign*}
 \At(X) \to \EK[\HybM](\R^n \times 2) 
\end{flalign*}
provides another Kleisli representation for the event-triggered programming
language
\[
    \prog[p]=\prog[a\in At(X) \mid skip\mid p\sComp p]
\]
\noindent
Its sequential composition now behaves essential like sequential composition
in \cite{hofner09}: discrete assignments are applied precisely at the end
of an evolution, and evolutions produced by two programs are concatenated.
\begin{example}[Stopwatch]
  Let $t$ be a variable that denotes time and consider the program
  composition
  \begin{flalign*}
    t : = 0 \prog[\sComp] (\dot{t} = 1 \> \& \> t = 5) \prog[\sComp]
    t := 0 \prog[\sComp] (\dot{t} = 1 \> \& \> t = 10)
  \end{flalign*}
  It yields the plot below.
  \begin{center}
  \scalebox{0.65}
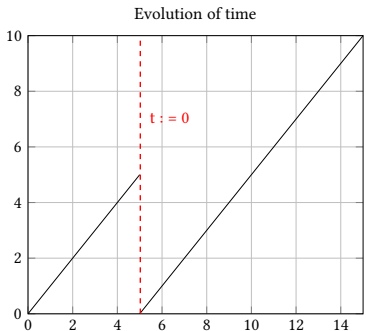
{
    \begin{tikzpicture}
      \begin{axis}[ title={Evolution of time}, ymin= 0, ymax = 10, xmin = 0, xmax = 15,
        grid = major]
        \addplot [smooth, domain = 0:5] { x };
        \addplot[smooth, domain = 0:15] { x - 5};
      \end{axis}
      \draw[line width=0.8, red,dashed] (2.3,0) -- (2.3,5.6);
      \node[] at (2.9,4) {\red{t : = 0}}; 
    \end{tikzpicture}
  }
\end{center}
\end{example}

\section{Definition of multiset-type monads}

Given a semiring $S$, the functorial part of the \emph{generalised
  multiset monad $\GMset:\Set\to\Set$} is defined by
\[
  \begin{cases}\GMset X= \left \{\phi: X\to S\mid  |\supp(\phi)|<\omega \right \}\\
    \GMset f:\GMset X\to\GMset Y, \hspace{0.2cm} \phi\mapsto \lambda
    y. \sum\limits_{x\in f\inv(\{y\})}\phi(x)\end{cases}
\]
The unit $\eta^{\GMset[]}$ is defined at each $X$ by
$\eta^{\GMset[]}_X(x)(y)=\delta_{x}(y)$, i.e. 1 if $x=y$ and 0
otherwise. The multiplication $\mu^{\GMset[]}$ is defined at each
$X$ by
$\mu^{\GMset[]}_X(\Phi)(x)=\sum_{\phi\in \supp(\Phi)}\Phi(\phi)\cdot
\phi(x)$, where $\cdot$ is the semiring multiplication.

\section{Proofs}

\begin{myProof}{Theorem~\ref{thm:TNfromInitCriterion}}
  For every $\cat$-object $X$ there exists a unique arrow
  $!_X: \emptyset\to X$, and thus any arrow $m: 1\to F\emptyset$ can
  be extended to a morphism $F!_X\circ m: 1\to FX$. The collection of
  all such morphisms forms a natural transformation
  $\alpha(m)$. Conversely, we can define
  $\phi: \left[\cat,\cat\right](\underline{1},F)\to
  \cat(1,F\emptyset), \beta\mapsto\beta_\emptyset$. It is clear that
  $\phi(\alpha(m))=m$, to see that
  $\alpha(\phi(\beta))=\alpha(\beta_\emptyset)=\beta$ use the
  naturality of $\beta$.
\end{myProof}

\begin{myProof}{Theorem~\ref{thm:PlusCoprodHom}}
  In $\Set$ binary products distribute over arbitrary coproducts. It
  follows that
  \begin{align*}
    F\times F=&(\coprod_{i\in I}\hom(X_i,-))\times (\coprod_{j\in I}\hom(X_j,-) )\\
    =& \coprod_{i,j\in I} \hom(X_i,-)\times\hom(X_j,-)= \coprod_{i,j\in I} \hom(X_i+X_j,-)
  \end{align*}
  where the last step follows from the fact that the contravariant
  $\hom$ functor sends colimits to limits. It follows that
  \begin{align*}
    \left[\Set,\Set\right](F\times F,F)&=\left[\Set,\Set\right] (\coprod_{i,j\in I} \hom(X_i+X_j,-),F )\\
                                       &=\prod_{i,j\in I}\left[\Set,\Set\right](\hom(X_i+X_j,-),F) \\
    & \cong \prod_{i,j\in I} F(X_i+X_j)
  \end{align*}
  where the last step is an application of the Yoneda lemma. Given an
  element $s\in\prod_{i,j\in I} F(X_i+X_j)$ with components
  $s_{i,j}\in F(X_i+X_j)$, the associated natural transformation
  $\alpha^s$ is defined at each $Y$ by:
  \begin{align}\label{eq1:thm:PlusCoprodHom}
    (a,b)\in\hom(X_i,Y)\times\hom(X_j,Y)\mapsto
    F\left[a,b\right](s_{ij})
  \end{align}
  where $[a,b]$ is the coproduct map $X_i+X_j\to Y$.
\end{myProof}

\begin{myProof}{Proposition \ref{prop:FinFuncRep}}
  We reason:
  \begin{flalign*}
    F \cong \pfs & \Lan(F_f) \\
    \stackrel{(1)}{\cong} \pfs & \Lan \left ( \colim \El[F_f]\op \stackrel{\Forg_{F_f}\op}{\longrightarrow}
      \omega\op \stackrel{\Yoneda}{\longrightarrow}[\omega, \Set] \right ) \\
    \stackrel{(2)}{\cong} \pfs & \colim \left ( \Lan \> \El[F_f]\op
      \stackrel{\Forg_{F_f}\op}{\longrightarrow}
      \omega\op \stackrel{\Yoneda}{\longrightarrow}[\omega, \Set] \right ) \\
    \stackrel{(3)}{\cong} \pfs & \colim \left ( \El[F_f]\op
      \stackrel{\Forg_{F_f}\op}{\longrightarrow}
      \omega\op \stackrel{\Inc\op}{\longrightarrow} \Set\op
      \stackrel{\Yoneda}{\longrightarrow}[\Set, \Set] \right )
  \end{flalign*}
  where $(1)$ is an application of Theorem~\ref{thm:FuncRep}, $(2)$
  takes advantage of $\Lan : [\omega,\Set] \to [\Set,\Set]$ being a
  left adjoint and $(3)$ uses the property
  $\Lan (\hom(n,-) : \omega \to \Set) \cong \hom(n, -) : \Set \to
  \Set$.
\end{myProof}

\begin{myProof}{Theorem \ref{thm:BinaryTNatFinitary}}
  We simply calculate:
  \begin{align*}
   & \left[\Set,\Set\right](F\times F,F)\\
    \stackrel{(1)} {=}&\left[\Set,\Set\right]\left(\colim \El[F_f\times F_f]\op \stackrel{\Yoneda\circ \Inc\op \circ \Forg\op}{\longrightarrow}\left[\Set,\Set\right],F\right)\\
    \stackrel{(2)}{=}&\lim \left[\Set,\Set\right]\left(\El[F_f\times F_f]\op \stackrel{\Yoneda\circ\Inc\op \circ \Forg\op}{\longrightarrow}\left[\Set,\Set\right],F\right)\\
    \stackrel{(3)}{=}&\lim \left(\El[F_f\times F_f] \stackrel{F_f\circ \Forg}{\longrightarrow}\Set\right)
  \end{align*}
  where $(1)$ is an application of Proposition \ref{prop:FinFuncRep},
  $(2)$ follows from the fact that the contravariant $\hom$ functor
  sends colimits to limits, and $(3)$ is an instance of the fact that
  for every functor $G : \catC \to \Set$ there exists an isomorphism
  \begin{flalign*}
    [\Set, \Set] \left (\Yoneda G\op, F \right ) \cong F G
  \end{flalign*}
  which follows from the Yoneda lemma.
\end{myProof}

\begin{myProof}{Theorem \ref{thm:GMsetTNat}}
For clarity we show the result for $n=2$ but the same argument holds
  for any finite $n$.  Let $\alpha: \GMset[S]\times\GMset[S]\to \GMset[S]$. 
  Since we only consider weight functions with finite support the functor $\GMset[S]$ is finitary  and therefore Theorem \ref{thm:BinaryTNatFinitary} applies. Since we're only considering multisets on the sets $n\in \omega$, we will write a multiset on $n$ simply as an $n$-tuple of elements of $S$. Note that when $\GMset[S]$ is applied to a map it creates
  a map between multisets which preserves the total mass of multisets. It follows that 
  $\mathbf{El}(\GMset[S]\times \GMset[S])$ has $S^2$ orbits since any pair $ \left((r_1,\ldots,r_n),(s_1,\ldots, s_n)\right)\in \GMset[S](n)\times \GMset[S](n)$ of multisets gets mapped to the pair  $(\sum_i^n r_i,\sum_i^n s_i)$ of their total mass under the map $!: n\to 1$.

  
  Let us now compute the limit \eqref{eq:TNatProdLim}. By definition, the orbits of $\mathbf{El}(\GMset[S]\times \GMset[S])$ cannot communicate with one another via morphisms, thus to compute the limit \eqref{eq:TNatProdLim} it is enough to characterize all the possible `threads' in the image under $\GMset[S]\circ\Forg$ of each orbit. Each of these threads correspond to a possible definition of the natural transformation $\alpha$ on the particular orbit in question. In particular, to specify $\alpha$ completely one need to describe what it does on each orbit. To do this we choose a pair of arbitrary multisets over $n$
  \begin{equation}\label{eq:mSetPair}
  \left((r_1,\ldots,r_n),(s_1,\ldots, s_n)\right).
  \end{equation}
  This choice fixes the orbit: we are now working in the orbit indexed by $(\sum_i^n r_i, \sum_i^n
  s_i)$ (and no morphism can make us jump to another orbit). It follows from the common integer divisor property that we can find a common integer divisor $q\in S$ for the finite sequence $(r_1,\ldots,r_n,s_1,\ldots,s_n)$ and a sequence $(l_1,\ldots,l_n,m_1,\ldots,m_n)$ with $r_i=l_i q, s_i=m_i q, 1\leq i\leq n$. Now let $M=\sum_i^n l_i$ and $N=\sum_i^n m_i$, assume w.l.o.g. that $M\leq N$ (else reverse the roles of $M$ and $N$ in what follows). We can write the pair of multisets \eqref{eq:mSetPair} as the image of the pair of multisets on $M+N$
  \begin{equation}\label{eq:mSetPairII}
  (\mu_1,\mu_2)=\left((\overbrace{q,\ldots,q}^M ,\overbrace{0,\ldots,0}^N),(\overbrace{0,\ldots,0}^M,\overbrace{q,\ldots,q}^N)\right)
  \end{equation}
  under the map $f: M+N\to n$ given by
  \[
  i\mapsto j \text{ if }\begin{cases} 1+\sum_{k=1}^{j-1} l_k\leq i\leq \sum_{k=1}^{j} l_k,\text{ or,}\\
  M+1+\sum_{k=1}^{j-1} m_k\leq i\leq M+\sum_{k=1}^{j} l_k
  \end{cases}
  \]
  The pair of multisets \eqref{eq:mSetPairII} is invariant under all permutations on $M+N $ of the shape $(\pi)(\rho)$ with $\pi\in\mathrm{Perm}(M)$ and $\rho\in\mathrm{Perm}(N)$. It follows by naturality that the image of this pair of multisets under $\alpha_{M+N}$ must also have this invariance property, and thus be of the shape
  \[
  (\overbrace{s,\ldots,s}^M,\overbrace{t,\ldots,t}^N)
  \]
  for some $s,t\in S$. Since $q$ is invertible we can write $s=q s' $ and $t=q t'$ for some $s',t'\in S$. By naturality we now have
  \begin{align*}
  &\alpha_n((r_1,\ldots,r_n),(s_1,\ldots, s_n))\\
  &=\alpha_n(f\times f(\mu_1,\mu_2))\\
  &= f(\alpha_{M+N}(\mu_1,\mu_2))\\
  &= f (q s',\ldots,q s', q t',\ldots,q t')\\
  &= (l_1 q s'+ m_1 q t',\ldots, l_n q s'+ m_n q t')\\
  &= (r_1 s' + s_1 t', \ldots, r_n s' + s_n t')
  \end{align*}
In other words, the image by $\alpha_n$ is necessarily a weighted sum of the multisets. Thus $(s',t')$ specifies an element in a thread of the limit \eqref{eq:TNatProdLim} associated the orbit indexed by $(\sum_i^n r_i, \sum_i^n s_i)$.

We now show that $(s',t')$ in fact specifies an entire thread uniquely by showing that, once chosen, this weighting must constant across the entire orbit. This is easily done: take any other pair of multisets in the same orbit $((r_1',\ldots,r_p'),(s_1',\ldots, s_p'))$, it follows from the common integer divisor property that we can find a common integer divisor $q'$ for all the elements in $((r_1',\ldots,r_p'),(s_1',\ldots, s_p'))$ \emph{and in} $((r_1,\ldots,r_n),(s_1,\ldots, s_n))$. By expressing all elements in terms of $q'$ we can apply the same trick as above and find that $\alpha_n$ and $\alpha_p$ will indeed take sums weighted by the same weight $(s',t')$.

We have thus shown that for a each choice $(\sum_i^n r_i,\sum_i^n s_i)$ of orbit, the corresponding elements of the limit \eqref{eq:TNatProdLim} are given by a weighting scheme $(s,t)\in S^2$. It follows that the limit \eqref{eq:TNatProdLim}, or equivalently the possible specifications of $\alpha$ across every orbit, is given by the set of maps $S^2\to S^2$ mapping a pair of total weights to a weighting scheme.

In other words, a map $\phi\colon S^2\to S^2$ defines the natural
transformation which sends a tuple of multisets with total masses
$(m_1,\ldots,m_n)$ to their sum weighted by the pair
$\phi(m_1,\ldots,m_n)$.
\end{myProof}

\begin{myProof}{Theorem \ref{thm:PlusGiry}}
 It follows from \eqref{eq:MachineGiry} that we need only consider the set 
 \[
 \left[\Pol_f,\Pol_f\right](\Giry_f\times\Giry_f,
  \Giry_f).
  \]
  It is not hard to see that $\Giry_f n$ is homeomorphic to the $n-1$-dimensional simplex with the usual topology and that the set of rational probabilities $\Giry_f^r n$ on $n$ forms a dense subset of the $n$-dimensional simplex. It is trivial to adapt the proof of Theorem \ref{thm:TNatRatDist} to show that the natural transformations $\Giry_f^r \times\Giry_f^r \to\Giry_f $ are given by the convex combinations $+^\lambda$, where $\lambda$ can now range over any values in $\unit$. It now follows from the fact that objects in $\Pol$ are complete that these transformations extend by continuity (since the $(n-1)$-dimensional simplex is compact) to give us all the natural transformations $\Giry_f\times\Giry_f\to\Giry_f$, and thus all the natural transformation $\Giry\times\Giry\to\Giry$.
\end{myProof}

\begin{myProof}{Theorem \ref{thm:plusSubDist}}
  It is not hard to check that the functor $\Giry\Maybe$ satisfies both the domain and codomain conditions of \cite{2016:Machine}. It follows that
\begin{align*}
  &\left[\Pol,\Pol\right](\Giry\Maybe\times \Giry\Maybe,\Giry\Maybe)\\
  &\cong \left[\Pol_f,\Pol_f\right]((\Giry\Maybe)_f\times (\Giry\Maybe)_f,(\Giry\Maybe)_f)
\end{align*}
  The space $(\Giry\Maybe)_f n=\Giry_f (n+1)$ is homeomorphic to the $n$-dimensional simplex with the usual topology and by completeness of Polish spaces it is enough to restrict ourselves to what happens to the dense subset of $\Giry^r_f
  (n+1)$ consisting of rational probabilities. The orbits of the category $\El[\Giry^r_f\Maybe \times\Giry_f^r\Maybe]$ are described by pairs of rationals $(r_1,r_2)\in\unit^2$ (describing the weight assigned to the `+1' component, or equivalently picking an $n-1$ hyperplane in each $n$-dimensional simplex). For a fixed orbit labelled by $(r_1,r_2)$ we can use the proof of Theorem \ref{thm:GMsetTNat} to show that the only possible natural assignments $\alpha_n:\Giry^r_f n+1\times \Giry^r_f n+1\to \Giry_f n+1$ are weighted sums of sub-distributions on $n$, i.e. weighted sums given by pairs $(q_1,q_2)\in\unit^2$ such that $q_1+q_2\leq 1$. However, each $\alpha_n$ must be a continuous map, and thus vary continuously across the orbits, and the conclusion follows.
  \end{myProof}
  
  \begin{myProof}{Theorem \ref{thm:PlusPow}}
  Consider any regular cardinal $\kappa>\lambda$, and the
  $\kappa$-accessible version $\Pow_\kappa$ of $\Pow$ (taking
  powersets of cardinality less that $\kappa$, see
  \cite{1994:AdamekLocally}). Theorem \ref{thm:BinaryTNatFinitary}
  generalises completely straightforwardly to $\kappa$-accessible
  functors, and
  $\left[\Set,\Set\right]((\Pow_\kappa)^\lambda, \Pow_\kappa)$ is thus
  given by the limit (\ref{eq:TNatProdLim}) (with the inclusion
  functor $\Inc$ suitably modified). To compute this limit, consider a
  set $X$ and a collection $(U_i)_{i\in \lambda}$ of subsets of
  $X$. It is easy to see, by considering what happens at the singleton
  1, that $\El[(\Pow_\kappa)^\lambda]$ has $2^\lambda$-orbits; one for
  each element of $(\Pow_\kappa 1)^\lambda$. The object
  $(X, (U_i)_{i\in \lambda})$ in $\El[(\Pow_\kappa)^\lambda]$ belongs
  to the orbit indexed by $(1,(!_X[U_i])_{i\in \lambda})$ where
  $!_X: X\to 1$, i.e. to the orbit determined by a subset of indices
  $J\in 2^\lambda$.

  We will show that any thread in the limit (\ref{eq:TNatProdLim})
  must pick an element $\bigcup_{k\in K} U_k$ in the copy of
  $\Pow_\kappa(X)$ corresponding to $(X, (U_i)_{i\in \lambda})$, for
  some subset of indices $K\in 2^\lambda$ such that $i\notin K$
  whenever $U_i=\emptyset$, and that this choice must be made
  consistently across the orbit. This will prove that a natural
  transformation $(\Pow_\kappa)^\lambda\to \Pow_\kappa$ is entirely
  determined by non-increasing maps $2^\lambda\to 2^\lambda$ (mapping $J$ to $K$).
  To prove the claim consider the object
  \begin{equation}\label{eq:objectPowkappa}
    \left(\biguplus_i U_i\times \lambda, \left(\biguplus_i U_i\times \{\epsilon_i\}\right)_{i\in \lambda}\right) \text{ in }\El[(\Pow_\kappa)^\lambda]
  \end{equation}
  where $\epsilon_i=i$ if $U_i\neq\emptyset$ and $\emptyset$ else.  It
  belongs to the same orbit as $(X, (U_i)_{i\in \lambda})$ as it is
  connected to it by the map $f: \biguplus_i U_i\times \lambda\to X$,
  defined by
  \[
    (u,i)\mapsto\begin{cases} u&\text{if }u\in U_i\\
      \text{any }u_i^0\in U_i& \text{if }	 u\notin U_i, U_i\neq\emptyset\\
      \text{anything }&\text{else}
    \end{cases}
  \]
  Note now that the object (\ref{eq:objectPowkappa}) is invariant
  under all endomorphisms which keep the $\lambda$-component
  constant. This means that for any thread in the limit
  (\ref{eq:TNatProdLim}), the component corresponding to the object
  (\ref{eq:objectPowkappa}) must contain
  $\biguplus_i U_i\times \epsilon_i$, i.e. that it will be a union
  \[
    \bigcup_{k\in K}\left(\biguplus_i U_i\times \epsilon_k\right)
  \]
  over some $K\subseteq \lambda$. Moreover, since
  $\epsilon_k=\emptyset$ when $U_k=\emptyset$ we do indeed have that
  $J\geq K$ (for the obvious order on $2^\lambda$). By pushing this
  component of the thread to the $(X, (U_i)_{i\in \lambda})$-component
  with $f$, we indeed get $\bigcup_{k\in K} U_k$ as claimed. It
  remains to check that the choice of $K\subseteq \lambda$ must be
  made consistently across the orbit. For this consider another object
  $(X, (V_i)_{i\in \lambda})$ in the same orbit (i.e. $V_i=\emptyset$
  iff $U_i=\emptyset$). We can always build the object
  \[
    (X^2,(U_i\times V_i)_{i\in \lambda})\text{ in
    }\El[(\Pow_\kappa)^\lambda]
  \]
  which gets mapped to $(X,(U_i)_{i\in \lambda})$ and
  $(X,(V_i)_{i\in I})$ by the projections maps
  $\pi_1,\pi_2: X^2\to X$. If the $(X,(U_i)_{i\in \lambda})$-component
  of a thread in the limit (\ref{eq:TNatProdLim}) is
  $\bigcup_{k\in K} U_k$, then the
  $(X^2,(U_i\times V_i)_{i\in \lambda})$-component of the thread must
  clearly be $\bigcup_{k\in K} U_k\times V_k$, and so the
  $(X,(V_i)_{i\in I})$-component must be $\bigcup_{k\in K} V_k$.

  Finally we need to show that our result for holds for the full
  powerset monad $\Pow$. Suppose for the sake of contradiction that
  $\alpha:(\Pow)^\lambda\to\Pow$ is not one of the transformations
  described above, then this must be witnessed at a set $X$,
  i.e. there must exist $(U_i)_{i\in \lambda}\in(\Pow X)^\lambda$ such
  that $\alpha_X((U_i)_{i\in \lambda})$ is not given by one of the
  transformations above. But since we can always find a regular
  cardinal such that $\Pow_\kappa X=\Pow X$, this would define a
  natural transformation
  $(\Pow_\kappa)^\lambda\times\Pow_\kappa \to\Pow_\kappa$ which is not
  of the form described above, a contradiction.
\end{myProof}

\noindent
Thus, intuitively, a non-increasing map $\phi : 2^n \to 2^n$ induces the natural
transformation $\Pow^n \to \Pow$ that given an $n$-tuple of subsets
$(X_1,\dots,X_n)$ with total masses $(m_1,\dots,m_n)$ it returns
$\cup \{ X_i \mid \pi_i \comp \phi(m_1,\dots,m_n) = 1 \}$.

\begin{myProof}{Proposition~\ref{prop:commutativity}}
  If a natural transformation $\alpha^s$ is commutative then for
  every $a : X_i \to X$, $b : X_j \to X$ the equation,
  \begin{flalign*}
    \alpha^s(a,b) = [a,b] \comp s_{ij} = \alpha^s(b,a) = [b,a] \comp
    s_{ji}
  \end{flalign*}
  must hold. In particular, the equation,
  \begin{flalign}
    \label{eq:comm}
    [i_2, i_1] \comp s_{ij} = [i_1,i_2] \comp s_{ji} = s_{ji}
  \end{flalign}
  holds for the injections $i_2 : X_i \to X_j + X_i$,
  $i_1 : X_j \to X_j + X_i$.

  Let us now assume that Equation~(\ref{eq:comm}) holds. The goal is
  to show that for every elements $a : X_i \to X$, $b : X_j \to X$ the
  equation $[a,b] \comp s_{ij} = [b,a] \comp s_{ji}$ holds. Thus reason,
  \begin{flalign*}
    [b,a] \comp s_{ji} = [b,a] \comp [i_2, i_1] \comp s_{ij} =
    [a,b] \comp s_{ij}
  \end{flalign*}
\end{myProof}

\begin{myProof}{Proposition~\ref{prop:idempotence}}
  If a transformation $\alpha^s$ is idempotent then for every
  $i \in I$, $a : X_i \to X$, the equation
  $\alpha^s(a,a) = [a,a] \comp s_{ii} = a : X_i \to X$ holds. In
  particular, we have
  \begin{flalign}
    \label{eq:idemp}
    \triangledown \comp s_{ii} =[\id,\id] \comp s_{ii} = \id : X_i \to
    X_i
  \end{flalign}
  Now assume that Equation~(\ref{eq:idemp}) holds and reason,
  \begin{flalign*}
    a = a \comp \id = a \comp [\id,\id] \comp s_{ii} = [a,a] \comp
    s_{ii}
  \end{flalign*}
\end{myProof}

\begin{proposition}\label{prop:units2}
  A natural transformation of the type $\const{1} \to F = \coprod_{i\in I}\hom(X_i,-)$
  factorises through some inclusion
  \begin{flalign*}
   \hom(X_k, -) \to F 
  \end{flalign*}
  with $k \in I$ and $X_k = \emptyset$.
\end{proposition}

\begin{proof}
  A natural transformation of the type
  $u : \const{1} \to F \cong \coprod_{i \in I} \hom(X_i, -)$ has in
  particular a map
  \begin{flalign*}
   u_\emptyset : 1 \to \coprod_{i \in I} \hom(X_i, \emptyset) 
  \end{flalign*}
  This entails the existence of an element $k \in I$ such that
  $X_k = \emptyset$ and
  $u_\emptyset(\ast) \in i_k(\hom(X_k, \emptyset))$. Using naturality
  and the map $a_X : \emptyset \to X$, it follows that for every set
  $X$, $u_X(\ast)$ lives in the $k$-th summand of
  $\coprod_{i \in I} \hom(X_i, X)$.
\end{proof}

\begin{myProof}{Proposition~\ref{prop:units}}
  If a natural transformation $\alpha^s : F \times F \to F$ has
  $u : \const{1} \to F$ as a unit then for every $a : X_i \to X$ the
  equation below holds
  \begin{flalign*}
    a = \alpha^s(a,x) = [a,x] \comp s_{ik} = \alpha^s(x,a) = [x,a] \comp s_{ki}
  \end{flalign*}
  where $x = u_X(\ast) : X_k = \emptyset \to X$. In particular, one has
  \begin{flalign*}
    \id = [\id,x] \comp s_{ik} = [x, \id] \comp s_{ki}
  \end{flalign*}
  Since $m = [\id,x] : X_i + \emptyset \to X_i$ and
  $n = [x, \id] : \emptyset + X_i \to X_i$ we obtain,
  \begin{flalign}
  \label{eq:units}
   \id = m \comp s_{ik} = n \comp s_{ki} 
  \end{flalign}
  Now assume that Equation~(\ref{eq:units}) holds and reason,
  \begin{flalign*}
    a = a \comp \id = a \comp m \comp s_{ik} = a \comp [\id,x] \comp s_{ik} =
    [a, x] \comp s_{ik}
  \end{flalign*}
  Using an analogous reasoning we obtain $a = [x,a] \comp s_{ki}$.
\end{myProof}

 \begin{myProof}{Proposition \ref{prop:absorption}}
   It follows by the naturality of $0: \underline{1}\to T$ applied to
   the map $\eta_X: X\to TX$ that $T\eta_X\circ 0_X = 0_{TX}$, and thus
   \begin{align}\label{eq:muSends0To0}
     \mu_X\circ 0_{TX}=\mu_X\circ T\eta_X\circ 0_X=0_X
   \end{align}
   We then have for every Kleisli representation $\lsem - \rsem$
   \begin{align*}
    \lsem 0 \sComp a \rsem &:=\mu_X\circ Ta\circ 0_X\circ !_X \\
    & =\mu_X\circ 0_{TX}\circ !_{X}&\text{Naturality of }0\\
    & =0_X\circ !_X &\text{ By Eq. }(\ref{eq:muSends0To0})\\
    & :=\lsem 0 \rsem
   \end{align*}
\end{myProof}

\begin{myProof}{Theorem \ref{theo_dist}}
  Assume that the axiom $p \prog[\sComp] 0 = 0$ does not hold and recall the
  bijections below.
 \begin{flalign*}
   [\Set, \Set](\const{1}, T) \cong [\Set,\Set](\hom(\emptyset, -), T)
   \cong T \emptyset
  \end{flalign*}
  We will show that $T \emptyset \not \cong 1$.
  By assumption, there exists an element $x \in X$ such that
  \begin{flalign*}
    \mu_X \comp T (0_X \comp !_X) \comp p (x) \not = 0_X \comp !_X (x).
  \end{flalign*}
  Hence, we have a natural transformation $\alpha : \const{1} \to T$ whose
  components are defined as
  \[
    \xymatrix@C=30pt{
      1 \ar[r]^{x} \ar@/_1.3pc/[rrrr]_{\alpha_X}
      & X \ar[r]^{p} & T X \ar[r]^(0.45){T (0_X \comp !_X)}
      & TT X \ar[r]^{\mu_X} & T X
    }
  \]
  Since $\alpha_X(\ast) \not = 0_X (\ast)$ the conditions
  $T \emptyset \not \cong 1$ holds.
\end{myProof}

\begin{myProof}{Theorem \ref{thm:PlotkinOpenProb}}
  We proceed by contradiction. Assume that there exists a unit
  transformation $\eta:\Id\to \Pow \Dist$ and a multiplication
  transformation $\mu:(\Pow \Dist)^2\to\Pow \Dist$ for which
  $\Pow \Dist$ is a monad. By definition of a monad we must have that
  at any $X$
  \begin{equation}\label{eq:PowfDist:rule1}
    \mu_X\circ \eta_{\Pow
      \Dist X}=\id_{\Pow
      \Dist X}
  \end{equation}
  and similarly,
  \begin{equation}\label{eq:PowfDist:rule2}
    \mu_X\circ \Pow
    \Dist\eta_X=\id_{\Pow\Dist X}.
  \end{equation}
  This set of equations gives us the action of $\mu_X$ on very
  specific inputs, and we will see that it is enough to generate a
  contradiction.

  As shown in Lemma \ref{lem:unitPowfDist}, $\eta$ is either the
  constant natural transformation to $\emptyset$ or is defined at
  $x\in X$ by $\eta_X(x)=\{\delta_x\}$. Assume first that $\eta$ is
  the constant natural transformation to $\emptyset$, and let $X$ be
  any non-empty set. Since the cardinality of $\Pow\Dist X$ is then at
  least 2, it is clear that there cannot exist a function
  $\mu_X: \Pow\Dist\Pow\Dist X\to\Pow\Dist X$ such that
  Eq. \ref{eq:PowfDist:rule1} holds for any $U\in \Pow\Dist X$
  :
  \[
    \mu_X\circ \eta_{\Pow\Dist X}(U)=\mu_X(\emptyset)=U.
  \]
  We therefore immediately get a contradiction if we assume that
  $\eta$ is the constant natural transformation to $\emptyset$.

  Next we assume that $\eta_X(x)=\{\delta_x\}$ and follow an argument
  due to Plotkin. Consider the sets $X=\{a,b,c,d\}$ and $Y=\{a,b\}$,
  the map $f: X\to Y$ defined by $f(a)=f(c)=a, f(b)=f(d)=b$ and the
  element
  $\{\frac{1}{2}\delta_{\{\delta_a,\delta_b\}}+\frac{1}{2}\delta_{\{\delta_c,\delta_d\}}\}\in
  \Pow \Dist\Pow \Dist X$. It is straightforward to compute that
  \begin{align*}
    \Pow
    \Dist\Pow
    \Dist f\left(\left\{\frac{1}{2}\delta_{\{\delta_a,\delta_b\}}+\frac{1}{2}\delta_{\{\delta_c,\delta_d\}}\right\}\right)& =\{\delta_{\{\delta_a,\delta_b\}}\}\\
                                                                                                                          &=\eta_{\Pow
                                                                                                                            \Dist Y}(\{\delta_a,\delta_b\})
  \end{align*}
  It now follows by naturality and Eq. \eqref{eq:PowfDist:rule1} that
  \[
    \xymatrix@C=4ex {
      \{\frac{1}{2}\delta_{\{\delta_a,\delta_b\}}+\frac{1}{2}\delta_{\{\delta_c,\delta_d\}}\}\ar@{|->}[r]\ar@{|->}[d]  & \mu_X(\{\frac{1}{2}\delta_{\{\delta_a,\delta_b\}}+\frac{1}{2}\delta_{\{\delta_c,\delta_d\}}\})\ar@{|->}[d]\\
      \{\delta_{\{\delta_a,\delta_b\}}\}\ar@{|->}[r] &
      \mu_Y(\eta_{\Pow \Dist
        Y}(\{\delta_a,\delta_b\}))=\{\delta_a,\delta_b\} }
  \]
  It follows that any distributions in
  $\mu_X(\{\frac{1}{2}\delta_{\{\delta_a,\delta_b\}}+\frac{1}{2}\delta_{\{\delta_c,\delta_d\}}\})$
  must belong to the preimage of $\{\delta_a,\delta_b\}$ under
  $\Pow \Dist f$, i.e. to
  \[
    \{p\delta_a+(1-p)\delta_c\mid p\in
    \unit\}\cup\{p\delta_b+(1-p)\delta_d\mid p\in \unit\}
  \]
  By considering the map $g: X\to Y$ defined by
  $g(a)=g(d)=a, g(b)=g(c)=b$, we also get
  \begin{align*}
    \Pow
    \Dist\Pow
    \Dist g\left(\left\{\frac{1}{2}\delta_{\{\delta_a,\delta_b\}}+\frac{1}{2}\delta_{\{\delta_c,\delta_d\}}\right\}\right)& =\{\delta_{\{\delta_a,\delta_b\}}\}\\
                                                                                                                          &=\eta_{\Pow
                                                                                                                            \Dist Y}(\{\delta_a,\delta_b\})
  \end{align*}
  and thus by naturality and Eq. (\ref{eq:PowfDist:rule1}), any
  distribution in
  \begin{flalign*}
   \mu_X(\{\frac{1}{2}\delta_{\{\delta_a,\delta_b\}}+\frac{1}{2}\delta_{\{\delta_c,\delta_d\}}\} 
  \end{flalign*}
  must also belong to the preimage of $\{\delta_a,\delta_b\}$ under
  $\Pow \Dist g$, i.e. to
  \[
    \{p\delta_a+(1-p)\delta_d\mid p\in
    \unit\}\cup\{p\delta_b+(1-p)\delta_c\mid p\in \unit\}
  \]
  It follows that
  $\mu_X(\{\frac{1}{2}\delta_{\{\delta_a,\delta_b\}}+\frac{1}{2}\delta_{\{\delta_c,\delta_d\}}\})$
  contains at most the elements
  \[
    \{\delta_a,\delta_b,\delta_c,\delta_d\}
  \]
  Now consider the map $h: X\to Z:=\{a,c\}$ defined by
  $h(a)=h(b)=a, h(c)=h(d)=c$. By definition we have
  \begin{align*}
    \Pow
    \Dist\Pow
    \Dist h\left(\left\{\frac{1}{2}\delta_{\{\delta_a,\delta_b\}}+\frac{1}{2}\delta_{\{\delta_c,\delta_d\}}\right\}\right)& \hspace{-2pt}=\hspace{-2pt}\left\{\frac{1}{2}\delta_{\{\delta_a\}}+\frac{1}{2}\delta_{\{\delta_c\}}\right\}\\
                                                                                                                          &\hspace{-2pt}=\hspace{-2pt}\Pow
                                                                                                                            \Dist\eta_Y(\{\frac{1}{2}\delta_a+\frac{1}{2}\delta_c\})
  \end{align*}
  It now follows by naturality and Eq. \eqref{eq:PowfDist:rule2} that
  \[
    \xymatrix@C=4ex {
      \{\frac{1}{2}\delta_{\{\delta_a,\delta_b\}}+\frac{1}{2}\delta_{\{\delta_c,\delta_d\}}\}\ar@{|->}[r]\ar@{|->}[d]  & \mu_X(\{\frac{1}{2}\delta_{\{\delta_a,\delta_b\}}+\frac{1}{2}\delta_{\{\delta_c,\delta_d\}}\})\ar@{|->}[d]\\
      \left\{\frac{1}{2}\delta_{\{\delta_a\}}+\frac{1}{2}\delta_{\{\delta_c\}}\right\}\ar@{|->}[r]
      &
      *\txt{$\mu_Y(\Pow
        \Dist\eta_Y(\{\frac{1}{2}\delta_a+\frac{1}{2}\delta_c\}))=$ \\
        $\{\frac{1}{2}\delta_a+\frac{1}{2}\delta_c\}$} }
  \]
  and we immediately get a contradiction since
  \[
  \Pow \Dist
  h(\{\delta_a,\delta_b,\delta_c,\delta_d\})=\{\delta_a,\delta_c\}
  \]
  and
  $\{\frac{1}{2}\delta_a+\frac{1}{2}\delta_c\}\notin
  \{\delta_a,\delta_c\}$.
\end{myProof}

\begin{myProof}{Theorem \ref{theo_QE}}
  It is straightforward to prove that
  $\delta : \HybM \NPow \to \NPow \HybM$ is a natural transformation.
  So we will now show that the natural transformation
  $\delta : \HybM \NPow \to \NPow \HybM$ makes the following diagram
  commute.
  \[
    \xymatrix{
      & \HybM \ar[dl]_{\HybM \eta^\NPow} \ar[dr]^{\eta^\NPow_\HybM} & \\
      \HybM \NPow \ar[rr]_{\delta} & & \NPow \HybM \\
      & \NPow \ar[ul]^{\eta^\HybM_\NPow} \ar[ur]_{\NPow \eta^\HybM} & }
  \]
  Start with the upper triangle.
  \begin{align*}
    & \delta_X \comp \HybM \eta_X^\NPow    \\
    = \pfs & \delta_X \comp ((\eta^\NPow \comp) \times \id) \\
    = \pfs & \eta^\NPow_{\HybM X}
  \end{align*}
  \noindent
  We then proceed with the lower one.
  \begin{align*}
    \pfs & \delta_X \comp \eta^\HybM_{\NPow X}  \\
    = \pfs & \left \{  \eta_X (a) \in \HybM X \mid a \in -
    \right \}  \\
    = \pfs & \NPow \eta^\HybM_X
  \end{align*}
  Finally, we will show that the natural transformation
  $\delta : \HybM \NPow \to \NPow \HybM$ makes the following diagram
  commute.
    \[
      \xymatrix{ \HybM \NPow \NPow \ar[r]^{\delta_\NPow} \ar[d]_{\HybM
          \mu} & \NPow \HybM \NPow \ar[r]^{\NPow \delta}
        & \NPow \NPow \HybM \ar[d]^{\mu_\NPow} \\
        \HybM \NPow \ar[rr]_{\delta} & & \NPow \HybM \\
        \HybM \HybM \NPow \ar[r]_{\HybM \delta} \ar[u]^{\mu_\NPow} &
        \HybM \NPow \HybM \ar[r]_{\delta_\HybM} & \NPow \HybM \HybM
        \ar[u]_{\NPow \mu} }
    \]
      Start with the upper square. Consider an element
  $(f,d) \in \HybM \NPow \NPow X$. A straightforward calculation
  provides the following equations.
  \begin{flalign*}
    \delta_X \comp \HybM \mu_X (f,d) 
    & = \left \{ \; (g,d) \in \HybM X
      \mid g \in \cup \comp f  \; \right \}   \\
    \mu_{\NPow X} \comp \NPow \delta_X \comp \delta_{\NPow X} (f,d) 
    & =  \bigcup \left \{ \; \delta_X \left (h,d \right ) \mid
      (h,d) \in \HybM \NPow X \wedge h \in f \; \right \} 
  \end{flalign*}
  \noindent
  We will show that both sets are indeed the same.  For this, start
  with an element $(g,d) \in \HybM X$, and reason in the following
  manner.
  \begin{align*}
    \pfs & g \in \cup \comp f  \\
    \Leftrightarrow \pfs & \forall t \in [0,\infty) . \> g \left (t \right ) \in \cup \comp f \\
    \Leftrightarrow \pfs & \forall t \in [0,\infty) . \> \exists Z_t \in f
    \left
      (t \right ) . \> g \left (t \right ) \in Z_t  \\
    \Leftrightarrow \pfs & \exists (h,d) \in \HybM \NPow X . \>
    g \in h \wedge h \in f  \\
    \Leftrightarrow \pfs & \exists (h,d) \in \HybM \NPow X . \> (g,d)
    \in \delta_X \left (h,d \right ) \wedge h \in f
  \end{align*}
  \noindent
  In order to keep the notation unburduned, and whenever no
  ambiguities arise, we will often use a pair $(f,d) \in \HybM X$ as
  if it were simply the map $f \in X^{[0,\infty)}$ which constantly
  outputs $f(d)$ after $d$ is achieved.

  Let us now concentrate on the lower square. Consider an element
  $(f,d) \in \HybM \HybM \NPow X$, and let $e = \pi_2(f (d))$. Then,
  a straightforward calculation shows that the following equations
  hold.
  \begin{flalign*}
    &\delta_X \comp \mu_{\NPow X} \left (f ,d \right )   \\
    & =  \left \{ \; (g,d + e) \in \HybM X \mid g \in
      (\theta_{\NPow X} \comp f, d)
      \conc (f (d)) \; \right \} \\
    & \NPow \mu_X \comp \delta_{\HybM X} \comp \HybM
    \delta_X \left (f , d \right )  \\
    & = \left \{ \; (\theta_X \comp g, d) \conc (g (d)) \mid
      (g,d) \in \HybM \HybM X \wedge g \in \delta_X \comp f \;
    \right \}
  \end{flalign*}
  \noindent
  We will show that both sets are actually the same.  Start with an
  element $(h,d + e) \in \HybM X$, and reason as follows.
  \begin{align*}
    \pfs & h \in (\theta_{\NPow X} \comp f, d) \conc (f (d))   \\[3pt]
    \Leftrightarrow \pfs & \forall t \leq d \; . \; \> h \left (t
    \right ) \in \theta_{\NPow X}
    \comp f \left (t \right ) \\ & \; \wedge \;
    \forall t > d \; . \; \> h \left (t \right ) \in (f (d)) \> (t - d)   \\[3pt]
    \stackrel{(\ast)}{\Leftrightarrow} \pfs & \forall t \leq d \; . \; \> h \left (t
    \right ) \in \NPow \theta_X \comp \delta_X \comp f \left (t \right
    ) \\ &  \; \wedge \;   \exists (g,e) \in
    \delta_X \left (f (d) \right ) \> . \>
    \forall t > d \; . \; \> h \left (t \right ) = g \left (t - d \right ) & \\[3pt]
    \Leftrightarrow \pfs & \exists (g,d) \in \HybM \HybM X \; . \;
    \forall t \leq d \; . \; g \left ( t \right ) \in \delta_X \left
      (f (t) \right ) \\ &
    \; \wedge \;  \theta_X \comp g \left (t \right ) = h \left
      ( t \right ) \; \wedge \; \forall t > d \; . \;  h \left ( t \right
    ) = (g (d)) \left (t - d \right )  & \\[3pt]
    \Leftrightarrow \pfs & \exists (g,d) \in \HybM \HybM X \; . \;
    (h, d + e) = (\theta_X \comp g, d) \conc (g (d)) \\ & \; \wedge \; g \in \delta_X \comp f &
  \end{align*}
\end{myProof}

\noindent
  Note that if instead of the functor $\NPow$ one would consider the
  powerset, then the equivalence $(\ast)$ would not hold. In
  particular, the equation
  \begin{flalign*}
    \theta_{\Pow X} \comp f = \Pow \theta_X \comp \delta_X \comp f
  \end{flalign*}
  would not necessarily hold.

\end{document}